\documentclass[a4paper,USenglish]{lipics}

\usepackage{ifthen}
\newcommand{\techRep}{true} %% switch here between true and false
\newcommand{\iftechrep}{\ifthenelse{\equal{\techRep}{true}}}

\usepackage{microtype}
\usepackage{amsmath}
\usepackage{amssymb}
\usepackage{color}
\usepackage{tikz}
\usetikzlibrary{arrows}

\usepackage{btran}
\theoremstyle{theorem}
\newtheorem{proposition}[theorem]{Proposition}
\theoremstyle{remark}
\newtheorem{claim}{Claim}
\newenvironment{qlemma}[1]{%
{\par\mbox{}\newline\noindent\bf Lemma #1.}%
\begin{itshape}%
}{%
\end{itshape}%
}

\newenvironment{qproposition}[1]{%
{\mbox{}\newline\noindent\bf Proposition #1.}
\begin{itshape}%
}{%
\end{itshape}%
}
\definecolor{orange3}{rgb}{1.0,0.2538,0.1681}
\definecolor{blau}{rgb}{0,.39608,0.74118}
\definecolor{rot}{rgb}{0.79216,.12941,0.24706}

\title{Stabilization of Branching Queueing Networks\footnote{S.~Kiefer is supported by a DAAD postdoctoral fellowship.}}
\titlerunning{Stabilization of Branching Queueing Networks} %optional, in case that the title is too long; the running title should fit into the top page column

\author[1]{Tom\'{a}\v{s} Br\'{a}zdil}
\author[2]{Stefan Kiefer}
\affil[1]{Faculty of Informatics, Masaryk University, Czech Republic\\
  \texttt{brazdil@fi.muni.cz}}
\affil[2]{Department of Computer Science, University of Oxford, United Kingdom\\
  \texttt{stefan.kiefer@cs.ox.ac.uk}}
\authorrunning{T. Br\'{a}zdil and S. Kiefer} %optional. First: Use abbreviated first/middle names. Second (only in severe cases): Use first author plus 'et. al.'

\Copyright[nc-nd]%choose "nd" or "nc-nd"
          {Tom\'{a}\v{s} Br\'{a}zdil and Stefan Kiefer}

% mandatory: Please choose ACM 1998 classifications from http://www.acm.org/about/class/ccs98-html . E.g., cite as "F.1.1 Models of Computation".
% mandatory: Please provide 1-5 keywords

\iftechrep{}{
\serieslogo{}%please provide filename (without suffix)
\volumeinfo%(easychair interface)
  {Billy Editor, Bill Editors}% editors
  {2}% number of editors: 1, 2, ....
  {Conference title on which this volume is based on}% event
  {1}% volume
  {1}% issue
  {1}% starting page number
\EventShortName{}
\DOI{10.4230/LIPIcs.xxx.yyy.p}% to be completed by the volume editor
}

\begin{document}
%---------------------
\maketitle

\sloppy

\newcommand{\N}{\mathbb{N}}
\newcommand{\Q}{\mathbb{Q}}
\newcommand{\Qp}{\mathbb{Q}_{\ge 0}}
\newcommand{\R}{\mathbb{R}}
\newcommand{\Rp}{\mathbb{R}_{\ge 0}}
\newcommand{\Z}{\mathbb{Z}}
\renewcommand{\P}[1]{{\cal P}\left(#1\right)}
\newcommand{\Prob}{\mathit{Prob}}
\newcommand{\Prod}[1]{\mathsf{R}^{(#1)}}
\newcommand{\Ex}[1]{\mathbb{E}\left[#1\right]}
\newcommand{\E}{\mathbb{E}}
\newcommand{\tim}{\mathbf{T}}
\newcommand{\norm}[1]{\left\| #1 \right\|}
\newcommand{\Net}{\mathcal{N}}
\newcommand{\tV}{\widetilde{V}}
\newcommand{\bSigma}{\bar{\Sigma}}

\newcommand{\va}{\boldsymbol{a}}
\newcommand{\valpha}{\boldsymbol{\alpha}}
\newcommand{\vb}{\boldsymbol{b}}
\newcommand{\vd}{\boldsymbol{d}}
\newcommand{\vDelta}{\boldsymbol{\Delta}}
\newcommand{\ve}{\boldsymbol{e}}
\newcommand{\vf}{\boldsymbol{f}}
\newcommand{\vlambda}{\boldsymbol{\lambda}}
\newcommand{\vmu}{\boldsymbol{\mu}}
\newcommand{\vq}{\boldsymbol{q}}
\newcommand{\vr}{\boldsymbol{r}}
\newcommand{\vu}{\boldsymbol{u}}
\newcommand{\vv}{\boldsymbol{v}}
\newcommand{\vw}{\boldsymbol{w}}
\newcommand{\vx}{\boldsymbol{x}}
\newcommand{\vy}{\boldsymbol{y}}
\newcommand{\vz}{\boldsymbol{z}}
\newcommand{\vzero}{\boldsymbol{0}}
\newcommand{\vone}{\boldsymbol{1}}

\newcommand{\as}[1]{\va^{(#1)}}
\newcommand{\asT}[1]{\va^{(#1)T}}
\newcommand{\es}[1]{\ve^{(#1)}}
\newcommand{\fs}[1]{\vf^{(#1)}}
\newcommand{\esT}[1]{\ve^{(#1)T}}
\newcommand{\ells}[1]{\ell^{(#1)}}
\newcommand{\qs}[1]{\vq^{(#1)}}
\newcommand{\qsT}[1]{\vq^{(#1)T}}
\newcommand{\ms}[1]{m^{(#1)}}
\newcommand{\xs}[1]{\vx^{(#1)}}

\begin{abstract}
 Queueing networks are gaining attraction for the performance analysis of parallel computer systems.
 A Jackson network is a set of interconnected servers, where the completion of a job at server~$i$ may result
  in the creation of a new job for server~$j$.
 We propose to extend Jackson networks by ``branching'' and by ``control'' features.
 Both extensions are new and substantially expand the modelling power of Jackson networks.
 On the other hand, the extensions raise computational questions, particularly concerning the stability of the networks,
  i.e, the ergodicity of the underlying Markov chain.
 We show for our extended model that it is decidable in polynomial time if there exists a controller that achieves stability.
 Moreover, if such a controller exists, one can efficiently compute a static randomized controller
  which stabilizes the network in a very strong sense; in particular, all moments of the queue sizes are finite.
\end{abstract}

\section{Introduction} \label{sec:intro}
Queueing theory plays a central role in the performance analysis of computer systems.
In particular, \emph{queueing networks} are gaining attraction as models of parallel systems.
A queueing network is a set of processing units (called \emph{servers}), each of which performs tasks (called \emph{jobs}) of a certain type.
Each server has its own queue of jobs waiting to be processed.
The successful completion of a job may trigger one (or more) new jobs (of possibly different type) that need to be processed as well.
In addition to this ``internal'' job creation, so-called \emph{open} queueing networks allow for new jobs to arrive ``externally'', i.e., from outside.

Queueing networks are a popular model for both hardware and software systems because of their simplicity and generality.
On the hardware side, queueing networks can, e.g., be used for modeling multi-core processors, see e.g.~\cite{Zisgen08}
 and the references in~\cite{DengPurvis11}.
One advantage of queueing-based analyses is their scalability with growing parallelism; e.g., it is said in~\cite{Madan11}:
``Cycle-accurate full-system performance simulators do not scale well beyond a few tens of processor cores at best.
As such, analytical models based on the theory of queueing systems, are a logical choice for developing a basic understanding
 of the fundamental tradeoffs in future, large-scale multi-core systems.''
On the software side, queueing networks are used for modeling message passing.
It is said in~\cite{LaTorre08}:
``Two natural classes of systems can be modeled using such a framework:
asynchronous programs on a multi-core computer and distributed programs communicating on a network.''
Of course, the realm of queueing networks stretches far beyond computer science, see \cite{Bolch,book:Daigle10}.

The simplest queueing networks are so-called \emph{Jackson networks}~\cite{Jackson57}:
Given two servers $i, j \in \{1, \ldots, n\}$, there is a ``rule'' of the form $i \btran{p_{ij}} j$ which specifies the probability $p_{ij}$
 that the completion of an $i$-job results in the creation of a $j$-job.
There are also rules $i \btran{p_{i0}} \varepsilon$ where $p_{i0} = 1 - \sum_{j} p_{ij}$ specifies the probability that no new job is created.
Each server~$i$ has a \emph{rate}~$\mu_i$ with which an $i$-job is processed if there is one.
In addition, there is a rate~$\alpha_i$ with which $i$-jobs arrive from outside the network.
The processing times and the external arrivals are exponentially distributed,
 so that a Jackson network describes a \emph{continuous-time Markov chain (CTMC)}.
It was shown in Jackson's paper~\cite{Jackson57} that
 if the rate $\lambda_i$ of internal \emph{and} external arrivals at server~$i$ is less than~$\mu_i$ for all~$i$,
  then the network is \emph{stable}, i.e., the average queue length is finite and almost surely all queues are empty infinitely often.
Moreover, Jackson networks allow a \emph{product form}, i.e., the steady-state distribution of the queue lengths can be expressed
 as a product of functions $\pi_i(k)$, where $\pi_i(k)$ is the steady-state probability that queue~$i$ has length~$k$.

\begin{example}[network processor] \label{ex:network-processor}
 In~\cite{AhmadiWong07}, Jackson networks are used to model network processors,
  i.e., chips that are specifically targeted at networking applications---think of a router.
 We describe the model from~\cite{AhmadiWong07} (sections 4.1 and~4.2, slightly adapted).
 Before packets are processed in the ``master processor''~$M$, they pass through the ``data plane''~$D$,
  from which a fraction~$q$ of packets needs to be processed first in the ``control plane''~$C$:
 \[
  D \btran{1-q} M \qquad D \btran{q} C \qquad C \btran{1} M \qquad M \btran{1} \varepsilon
 \]
 An ``arrival manager''~$A$ sends some packets (fraction~$d_0$) directly to~$D$,
  but others (fractions~$d_1,\ldots,d_n$ with $d_0 + d_1 + \cdots + d_n=1$) are sent to ``slave processors'' $S_1, \ldots, S_n$ to assist the master.
 Some packets (fraction~$b$) still need the attention of the master after having been processed by a slave:
 \begin{equation}
  A \btran{d_0} D \qquad A \btran{d_i} S_i \qquad S_i \btran{b} D \qquad S_i \btran{1-b} \varepsilon \quad , \quad i \in \{1, \ldots, n\}. \tag*{}
 \end{equation}
\end{example}

Jackson networks and their extensions have been thoroughly studied, but they are restricted in their modelling capabilities,
 as (i) the completion of a job may trigger at most one job,
 and (ii) there is no nondeterminism that would allow to control the output probabilities of a server.
Considering~(i), it seems unnatural to assume that a distributed program communicating on a network
 produces at most one message at the end of its computation.
Considering~(ii), the ``arrival manager''~$A$ in Example~\ref{ex:network-processor}
 may want to flexibly pass incoming packets to the master or one of the slaves, possibly depending on the current load.
%a system designer may want to flexibly assign newly created jobs to different processor cores, possibly depending on the current load;
These restrictions have not been fully addressed, not even in isolation.
In this paper we introduce \emph{controlled branching queueing networks}, which are Jackson-like networks but allow for
 both nondeterminism (``controlled'') and the creation of more than one job (``branching'').

Both extensions directly raise computational issues.
We show in Example~\ref{ex:no-product-form} on page~\pageref{ex:no-product-form} that even purely stochastic branching networks do not allow a product form,
 which illustrates the mathematical challenge%
\footnote{It is noted in~\cite{HarrisonWilliams92} that ``[\ldots] virtually all of the models that have been successfully analyzed
 in classical queueing network theory are models having a so-called product form stationary distribution.''}
 of this extension
  and poses the question for an effective criterion that allows to determine whether the network is stable,
 i.e., returns to the empty state infinitely often.
Moreover, due to the nondeterminism, we now deal with \emph{continuous-time Markov decision processes (CTMDPs)}.
Our main theorem (Theorem~\ref{thm:main-charact}) states that if there exists any scheduler
 resolving the nondeterminism in such a way that the controlled branching network is stable,
 then there exists a \emph{randomized static} scheduler that achieves stability as well,
 where by ``randomized static'' we mean that the decisions may be randomized but not dependent on the current state (such as the load) of the system.
Moreover, the existence of such a stabilizing scheduler and the scheduler itself can be determined in polynomial time,
 and, finally, the randomized static scheduler is stabilizing in a very strong sense, in particular, all moments of the queue sizes are finite.

\emph{Related work.}
We use nondeterminism to describe systems whose behaviour is not completely specified.
A system \emph{designer} can then resolve the nondeterminism to achieve certain goals, in our case stability.
Although nondeterminism is a very well established modelling feature of probabilistic systems (see e.g.~\cite{Kwiat11}),
 the literature on automatic design of stabilizing controllers for queueing networks is sparse.
\emph{Flow-controlled} networks \cite{Stidham85,Massey91} allow to control only the external arrival stream or the service rates
 (see also \cite{Azaron03} and the references therein).
The authors of~\cite{Kumar96,Glazebrook99} consider queueing networks with fewer servers than job types,
 so that the controller needs to assign servers to queues.
As in~\cite{Kumar96,Glazebrook99}, we also use linear programming to design a controller,
 but our aim is different:
 we allow the controller to influence the production of the individual queues,
 and we study the complexity of designing stabilizing controllers and the nature of such controllers.
There has been a substantial amount of work in the last years analyzing probabilistic systems with ``branching features'',
 most prominently on \emph{recursive Markov chains} \cite{EYstacs05Extended,EtessamiWY08} and \emph{probabilistic pushdown systems} \cite{EKM04,BKKV:ICALP11}.
While these models allow for a probabilistic splitting of tasks by pushing new procedures on a stack,
 the produced tasks are processed in a strictly sequential manner, whereas the queues in a queueing network process jobs in parallel and in continuous time.
Recently, \emph{probabilistic split-join systems} were introduced~\cite{KW11:TACAS}, which allow for branching but not for external arrivals,
 and assume unlimited parallelism.
In~\cite[chapter~8]{KitaevRykov95} a queueing model with multiple classes of tasks and ``feedback'' is discussed,
 which is similar to our branching except that there is only one server, hence there is no parallelism.
%Hence there is no parallelism; instead, the server needs a policy to decide which of the competing tasks is allotted service time.
%
Algorithmic theory of queueing systems has also attracted some attention in the past.
In particular, for \emph{closed} (i.e., without external arrivals) queueing systems,
 \cite{Papadimitriou94thecomplexity} shows EXP-completeness of minimizing a weighted throughput of the queues.
%For open networks, this problem would be somewhat stronger than our problem of existence of stable strategies (which does not make sense for closed networks).

\newcommand{\run}{\mathit{Run}}
\newcommand{\Init}{\mathit{In}}

\section{Preliminaries} \label{sec:prel}

\emph{Numbers.}
We use $\Z, \Q, \R$ for the sets of integer, rational, real numbers, respectively, and $\N, \Qp, \Rp$ for their respective subsets of nonnegative numbers.

\smallskip\noindent
\emph{Vectors and Matrices.}
Let $n \ge 1$.
We use boldface letters for denoting vectors $\vx = (\vx_1, \ldots, \vx_n) \in \R^n$.
Vectors are row vectors per default, and we use superscript $^T$ for transpose, so that $\vx^T$ denotes a column vector.
If the dimension $n$ is clear from the context, we write $\vzero := (0, \ldots, 0)$, $\vone := (1, \ldots, 1)$, and
 $\es{i} = (0, \ldots, 0, 1, 0, \ldots, 0)$ for the vector with the $1$ at the $i$th component ($1 \le i \le n$).
It is convenient to define $\es{0} := \vzero$.
For two vectors $\vx, \vy \in \R^n$ we write $\vx \sim \vy$ with $\mathord{\sim} \in \{\mathord{=},\mathord{<},\mathord{\le},\mathord{>},\mathord{\ge}\}$
 if the respective relation holds componentwise.
For a vector $\vx \in \R^n$ we denote its \emph{1-norm} by~$\norm{\vx} := \sum_{i=1}^n |\vx_i|$.
When $\vx \in \N^n$ is a vector of queue sizes, we refer to~$\norm{\vx}$ as the \emph{total queue size}.
For a matrix $A \in \R^{n \times n}$, we write $A_i$ for its $i$th row, i.e., $A_i = (A_{i 1}, \ldots, A_{i n})$.
%The spectral radius of a matrix~$A$ is defined as the largest absolute value of its eigenvalues.
%It is known~\cite{book:BermanP}
% that the spectral radius of a nonnegative matrix~$A \in \Rp^{n \times n}$ is less than~$1$,
% if and only if the matrix series $A^* := \sum_{i=0}^\infty A^i$ converges (``exists'') in~$\Rp^{n \times n}$.
%In this case we have $A^* = (I-A)^{-1}$, where $I$ denotes the identity matrix.

\smallskip\noindent
\emph{CTMDP.} A {\em continuous-time Markov decision process (CTMDP)} consists of an at most countable set $S$ of states,
 an initial state~$s_1 \in S$,
 a~set of actions $\Sigma$~\footnote{Usually, each state has its own set of available actions.
 As this feature is not needed for queueing networks, we stick to the simpler version in which all actions are always available.},
 and a transition rate $q(s,\sigma,s')\geq 0$ for each pair of states $s,s'\in S$ and each action $\sigma\in \Sigma$ (here
 $q(s,\sigma,s')=0$ means that the transition from $s$ to $s'$ never occurs).
We define a {\em continuous-time Markov chain (CTMC)} to be a CTMDP whose set of actions $\Sigma$ is a singleton
 (we usually do not write the only action explicitly, so the transition rates of a CTMC are denoted by $q(s,s')$, etc.).

Intuitively, a run of a CTMDP starts in~$s_1$ % a randomly chosen state according to the initial distribution $\Init$.
 and then evolves in so-called epochs.
Assume that (after the previous epoch) the system is in a state $s$.
The next epoch consists of the following phases: First, a scheduler chooses an action $\sigma\in \Sigma$ to be executed.
Second, a waiting time for transition to each state $s'\in S$ is chosen according to the exponential distribution with the rate $q(s,\sigma,s')$
 (here we assume that if $q(s,\sigma,s')=0$, then the waiting time is~$\infty$).
The transitions compete in a way such that the one with the least waiting time is executed and the state of the CTMDP is chosen accordingly
 (the other transitions are subsequently discarded).
%\tomas{Removed the sentence about the branching probability.}
%As the waiting times are exponentially distributed, the probability that a transition from $s$ to $t$ wins the competition under action~$\sigma$ is given by
%$
%P(s,\sigma,t)=\frac{q(s,\sigma,t)}{\sum_{t'\in S} q(s,\sigma,t')}\,.
%$

Formally, a~{\em run} is
an infinite sequence $s_1,\sigma_1,t_1,s_2,\sigma_2,t_2,\ldots\in (S\times \Sigma\times \Rp)^{\omega}$. We denote by $\run$ the set of all runs.
A {\em scheduler} is a function $\Theta$ which assigns to every finite path $s_1,\sigma_1,t_1,s_2,\sigma_2,t_2,\ldots s_{n}\in (S \times \Sigma\times \Rp)^*\times S$
 a probability distribution~$\Theta(w)$ on actions (i.e. $\Theta(w):\Sigma\rightarrow [0,1]$ satisfies $\sum_{\sigma\in \Sigma} \Theta(w)(\sigma)=1$).
%\stefan{Tomas, please add the reference.}
For technical reasons, we have to restrict ourselves to measurable schedulers (for details see~e.g.~\cite{NSK:CTMDP-delayed}).

We work with a measurable space of runs $(\run,\mathcal{F})$ where $\mathcal{F}$ is the smallest $\sigma$-algebra generated
by basic cylinders (i.e. sets of runs with common finite prefix) in a standard way.
%\tomas{This is not completely formally precise, as in continuous time the prefixes may have different waiting times from a fixed interval. Stefan, what do you think?}
%\stefan{I'm ok.}
Every
% of the form $\prod_{i=1}^{n-1} (\{s_i\}\times \{\sigma_i\}\times T_i))\times \{s_{n}\}$ where for every $i\in \{1,\ldots,n\}$ we have that
% $s_i\in S$,
%%$\sigma_i\subseteq \Sigma$
%\stefan{Tomas, please check if I correctly corrected $\sigma_i \subseteq \Sigma$ to $\sigma_i \in \Sigma$.}
%$\sigma_i \in \Sigma$
%and $T_i$ is a compact subinterval of $\Rp$.
scheduler $\Theta$ induces a unique probability measure $\Pr_{\Theta}$ on $\mathcal{F}$ determined by the probabilities of the basic cylinders. %cylinder sets.
For detailed definitions see~\cite{NSK:CTMDP-delayed}.
%More concretely, $\Pr(\prod_{i=1}^{n-1} (\{s_i\}\times \{\sigma_i\}\times T_i))\times \{s_{n}\})$ is defined to be the probability that
% %$s_0$ is the initial state and that
% in the $i$-th epoch, for all $i\in \{1,\ldots,n\}$, the process visits $s_i$, executes $\sigma_i$, and waits a random time which belongs to $T_i$
%\stefan{Tomas, please add the reference.}
%  (for details see~e.g.~\cite{--}).
%\tomas{incorrect!!!!!}
%\begin{align*}
%{\Pr}_{\Theta} & (\prod_{i=0}^{n-1} (\{s_i\} \times \{\sigma_i\}\times T_i)\times \{s_{n}\})= \\ & \Init(s_0) \cdot
%\int_{a_0}^{b_0}\cdots \int_{a_{n-1}}^{b_{n-1}} \prod_{i=0}^{n-1}\left( \sum_{\zeta\in \Sigma}  \Theta(s_0,\sigma_0,t_0,\ldots,s_i)(\zeta)\cdot P(s_i,\zeta,s_{i+1}) \mathop{\cdot}
%r_i\cdti e^{-r_i t_i}\right)\ d t_{n-1} \cdots d t_0\\
%& \text{ where } r_i=q(s_i,\sigma_{i},s_{i+1}) \text{ for } i\in \{1,\ldots,n\}
%\end{align*}
%\tomas{Removed the deterministic scheduler and the sentence about induced CTMC.}
Then each scheduler $\Theta$ {\em induces} a stochastic process $\left(x(t)\mid t\in \Rp\right)$ on the probability space $(\run,\mathcal{F},\Pr_{\Theta})$ where $x(t)$ is the current state of the run in time $t$, i.e.,
each $x(t)$ is a random variable defined by
\[
x(t)(s_1,\sigma_1,t_1,s_2,\sigma_2,t_2,\ldots)=s_i \quad , \quad \sum_{j=1}^{i-1} t_i\leq t \text{ and } \sum_{j=1}^{i} t_i\geq t\,.
\]
%\tomas{The action of CTMC is, of course, usually omited from runs, rates, etc. We should probably do the same here}
A scheduler $\Theta$ is {\em memoryless} if for every path $w=s_1,\sigma_1,t_1,s_2,\sigma_2,t_2,\ldots s_{n+1}\in (S \times \Sigma\times \Rp)^*\times S$ we have that $\Theta(w)=\Theta(s_{n+1})$.
%
%and $\Theta$ is {\em deterministic}
%\stefan{get rid of deterministic?}
%if for every path $w$ the distribution $\Theta(w)$ is Dirac (i.e. $\Theta(w)(\sigma)=1$ for some $\sigma\in \Sigma$).
%Note that by fixing a memoryless scheduler $\Theta$ we obtain a CTMC with the following rates:
%\[
%q(s,s')=\sum_{\sigma \in \Sigma} \Theta(s)(\sigma)\cdot q(s,\sigma,s')\,.
%\]

\smallskip\noindent
\emph{Networks.}
Define $\Prod{n,K} := \{\vr \in \N^n \mid \vr_1 + \cdots + \vr_n \le K\}$.
A \emph{production function} for $(n,K)$ is a function $\Prob: R \to \Q \cap (0,1]$ with $R \subseteq \Prod{n,K}$
 such that $\sum_{\vr \in R} \Prob(\vr) = 1$.
A \emph{controlled branching network} \emph{with $n$ queues and branching factor~$K$}
 consists of an arrival rate $\mu_0 \in \Q \cap (0,\infty)$,
 queue rates $\mu_i \in \Q \cap (0,\infty)$ for $i \in \{1, \ldots, n\}$,
 an arrival production function $\Prob_0: R_0 \to \Q \cap (0,1]$ for $(n,K)$,
 and finite action sets $\Sigma_1, \ldots, \Sigma_n$ as follows.
An action $\sigma_i \in \Sigma_i$ assigns to queue~$i$ a production function $\Prob_i(\sigma_i): R_i(\sigma_i) \to \Q \cap (0,1]$ for $(n,K)$.
Define $\Sigma := \Sigma_1 \times \cdots \times \Sigma_n$.
If $\sigma = (\sigma_1, \ldots, \sigma_n) \in \Sigma$,
 we write $R_i(\sigma)$, $\Prob_i(\sigma)$ and $R_0(\sigma)$, $\Prob_0(\sigma)$
  to mean $R_i(\sigma_i)$, $\Prob_i(\sigma_i)$ and $R_0$, $\Prob_0$.
Observe that the rates~$\mu_i$ do not depend on actions.
This simplification is without loss of generality.%
\footnote{To show that this assumption is w.l.o.g.\ one can employ the standard ``uniformization'' trick.
More precisely, assume that the actions of a queue~$i$ have different rates.
Define $\mu_i$ to be the maximum of all rates of~$\Sigma_i$ and compensate by ``adding self-loops'', i.e.,
 make the actions of~$\Sigma_i$ generate a new job for queue~$i$ with a suitable probability.
This effectively substitutes a transition with longer delay by possibly several transitions with delay~$\mu_i$.
As static schedulers can be easily translated between the original and the transformed system, our results remain valid.}
%\tomas{Stefan, please check my rather vague explanation.... is it ok?}
%\stefan{Yes. I have changed it a little.}
%~\footnote{To show that this assumption is w.l.o.g. one can employ the standard ``uniformization'' trick. More precisely, assume that a queue $i$ does not have the uniform rate.
%Define $\mu_i$ to be the maximum of all rates of $\Sigma_i$ and compensate by making each action of $\Sigma_i$ generate a new job for queue $i$ with an appropriate probability (this causes that one transition with longer delay is substituted with possibly several transitions with the delay $\mu_i$). As static schedulers can be easily translated between the original and the transformed system, our results about stabilizing networks remain valid.}.
%\tomas{Stefan, please check my rather vague explanation.... is it ok?}
We assume a nonzero arrival stream, i.e., there is $\vr \in R_0$ with $\vr \ne \vzero$.
We define the \emph{size} of a controlled network by
 $n + K + \left(\sum_{i=0}^n |\mu_i|\right) + |R_0| + |\Prob_0| + \sum_{i=1}^n \sum_{\sigma_i \in \Sigma_i}  |R_i(\sigma_i)| + |\Prob_i(\sigma_i)|$,
 where $|\mu_i|$ etc.\ means the description size assuming the rationals are represented as fractions of integers in binary.
A controlled branching network induces a CTMDP
 with state space~$\N^n$ (the queue sizes), initial state~$\vzero$, action set~$\Sigma$, and transition rates
% \[
%  q(\vx,\sigma,\vy) = \sum_{\vr \in R_0(\sigma) : \vy = \vx + \vr} \mu_0 \Prob_0(\vr) \quad + \quad
%  \sum_{i = 1}^n \ \sum_{\vr \in R_i(\sigma) : \vy = \vx - \es{i} + \vr} \mu_i \Prob_i(\vr)
%   \qquad \text{for $\vx,\vy \in \N^n$, $\sigma \in \Sigma$.}
% \]
 \[
  q(\vx,\sigma,\vy) = \sum_{i\in \{0,1,\ldots,n\}: i=0 \lor \vx_i \ne 0} \ \sum_{\vr \in R_i(\sigma) : \vy = \vx - \es{i} + \vr} \mu_i \Prob_i(\sigma)(\vr)
   \qquad \text{for $\vx,\vy \in \N^n$, $\sigma \in \Sigma$.}
 \]
Interpreting this definition, there is a ``race'' between external arrivals (rate~$\mu_0$) and the nonempty queues (rates~$\mu_i$);
 if the external arrivals win, new jobs are added according to~$\Prob_0(\sigma)$;
 if queue~$i$ wins, one $i$-job is removed and new jobs are added according to~$\Prob_i(\sigma)$.

An \emph{purely stochastic branching network} is a controlled branching network with $\Sigma = \{\sigma\}$,
 i.e., with a unique action for each queue.
Hence, the induced CTMDP is a CTMC.
In the purely stochastic case we write only $R_i$ for $R_i(\sigma)$ etc.
If $\Prob_i(\vr) = p$ in the purely stochastic case, we use in examples the notation $i \btran{p} \vr$,
 where we often write $\vr \in \Prod{n,K}$ as a multiset without curly brackets.
For instance, if $n=2$, we write $1 \btran{p} 1,2$ and $1 \btran{p} 2,2$ and $1 \btran{p} \varepsilon$
 to mean $\Prob_1(\vr) = p$ with $\vr = (1,1)$ and $\vr = (0,2)$ and $\vr = (0,0)$, respectively.

Fixing a controlled network~$\Net$ and a scheduler $\Theta$ for the CTMDP induced by~$\Net$, we obtain a stochastic process
$\Net_{\Theta}=(\vx(t) \mid t\in \Rp)$, where $\vx(0)=\vzero\in \N^n$, which evolves according to the dynamics of~$\Net$ and the scheduler $\Theta$.
In the purely stochastic case we drop the subscript~$\Theta$, and so we identify a network~$\Net$ with its induced stochastic process.

\begin{example}[no product form] \label{ex:no-product-form}
 Consider the purely stochastic branching network with $0 \btran{1} 1,2$ and $1 \btran{1} \varepsilon$ and $2 \btran{1} \varepsilon$.
 If its stationary distribution~$\pi$ (for a definition of stationary distribution see before Theorem~\ref{thm:main-charact})
 had product from, the queues would be ``independent in steady-state'', i.e.,
  $\pi(\vx_2 \ge 1 \mid \vx_1 \ge 1) = \pi(\vx_2 \ge 1)$, where by~$\vx$ we mean $\vx(t)$ in steady state.
 However, if $\mu_0$ is much smaller than $\mu_1=\mu_2$, then we have $\pi(\vx_2 \ge 1 \mid \vx_1 \ge 1) > \pi(\vx_2 \ge 1)$,
  intuitively because $\vx_1 \ge 1$ probably means that there was an arrival recently, so that $\vx_2 \ge 1$ is more likely than usual.
 More concretely, let $\mu_0 = 1$ and $\mu_1 = \mu_2 = 3$
  and consider the 2-state Markov chain obtained by assuming that each arrival leads to the state $(1,1)$ and each completion of any job leads
  to the state $(0,0)$.
 By computing the stationary distribution $\pi'$ of this 2-state Markov chain in the standard way,
  we obtain $\pi'((0,0)) = 6/7$ and $\pi'((1,1)) = 1/7$.
 Since this 2-state Markov chain ``underapproximates'' the CTMC induced by the network,
  we have $\pi(\vx_1 \ge 1 \ \land \ \vx_2 \ge 1) \ge 1/7$.
 On the other hand, by considering the two queues separately, the standard formula for the M/M/1 queue gives
  $\pi(\vx_1 \ge 1) = \pi(\vx_2 \ge 1) = 1/3$.
 Product form would imply $\pi(\vx_1 \ge 1 \ \land \ \vx_2 \ge 1) = \pi(\vx_1 \ge 1) \cdot \pi(\vx_2 \ge 1) = 1/9$,
  contradicting the inequality above.
\end{example}

\section{Results} \label{sec:results}
%In this section we present the results of this paper.
We focus on the stability of purely stochastic and controlled branching networks.
Our notion of stability requires that the network is completely empty infinitely many times.
%Given a controlled branching network~$\Net$,
% we say that a scheduler $\Theta$ is {\em ergodic for~$\Net$} if the expected return time of~$\Net_{\Theta}$ to $\vzero$ is finite.
%More formally, define a random variable $R_{\Net,\Theta}$ by
%\[
%R_{\Net,\Theta}\quad :=\quad \inf_{t>0}\, \left\{t\mid \vx(t)=\vzero, \exists t'<t:\vx(t')\not = \vzero\right\} \,.
%\]
%Then $\Theta$ is ergodic for $\Net$ iff $\Ex{R_{\Net,\Theta}}<\infty$.
%\stefan{I have changed the definition of ergodic. See also the end of preliminaries. Old version is commented.}
Given a stochastic process $(\vx(t) \mid t\in \Rp)$,
 we say that the process is \emph{ergodic} if the expected return time to $\vzero$ is finite.
More formally, define a random variable $R$ by
\[
R\quad :=\quad \inf_{t>0}\, \left\{t\mid \vx(t)=\vzero, \exists t'<t:\vx(t')\not = \vzero\right\} \,.
\]
Then the process is ergodic iff $\Ex{R}<\infty$.
In the controlled case, we say that a scheduler~$\Theta$ for~$\Net$ is \emph{ergodic for~$\Net$} if $\Net_\Theta$ is ergodic.
In the following we use stability and ergodicity interchangeably.
A scheduler $\Theta$ is {\em static} if it always chooses the same fixed distribution on actions.
Note that static schedulers are memoryless.
%We call a purely stochastic branching network~$\Net$ ergodic if its unique scheduler is ergodic for~$\Net$.
%
If in a stochastic process $(\vx(t) \mid t\in \Rp)$ the limit $\pi(\vx) := \lim_{t \to \infty} \Pr(\vx(t) = \vx)$ exists for all $\vx \in \N^n$ and $\sum_{\vx\in \N^n} \pi(\vx)=1$,
 then $\pi: \Rp \to [0,1]$ is called the \emph{stationary distribution}.
%
%\begin{theorem}
%Let $\Net$ be a network.
%The following conditions are equivalent:
%\begin{enumerate}
%\item There is a {\em static randomized} ergodic scheduler for $\Net$.
%\item There is an (arbitrary) ergodic scheduler for $\Net$.
%\item There is a scheduler and a constant $c\in \Rp$ such that
%\[
%\Pr\left(\lim_{t\rightarrow \infty} \frac{\int_{0}^t s(t) dt}{t}<c\right)=1
%\]
%where $s(t)$ equals the sum of all components of $\vx(t)$, i.e., to the sum of sizes of all queues in time $t$.
%\end{enumerate}
%\end{theorem}
%%
%
%\begin{theorem}\label{thm:main-charact}
%Let $\Net$ be a network and assume that there is an (arbitrary) ergodic scheduler for~$\Net$.
%Then there is a {\em static randomized} ergodic scheduler for~$\Net$
%%such that
%%\begin{itemize}
%%\item $\P{R_{\Net,\Theta}\leq t}\quad \leq \quad .... \text{ something exponentially small :-) }$
%%\item There is a constant $s_{\mathit{total}}\leq ...$ such that for almost all runs of~$\Net_{\Theta}$
%% the average total size of queues is equal to $s_{\mathit{total}}$.
%%\end{itemize}
%with invariant distribution~$\pi$ such that there exists an exponential moment of the total number of waiting jobs,
% i.e., there is $\delta > 0$ such that
% $\sum_{\vx \in \N^n} \exp(\delta \norm{\vx}_1) \pi(\vx)$ exists.
%It is decidable in polynomial time whether there exists an arbitrary ergodic scheduler,
% and if it exists, a static randomized ergodic scheduler as above can be computed in polynomial time.
%\end{theorem}
%
\begin{theorem}\label{thm:main-charact}
Let $\Net$ be a controlled branching network.
It is decidable in polynomial time whether there exists an (arbitrary) ergodic scheduler for~$\Net$.
If it exists, one can compute, in polynomial time, a {\em static randomized} ergodic scheduler for~$\Net$
%such that
%\begin{itemize}
%\item $\P{R_{\Net,\Theta}\leq t}\quad \leq \quad .... \text{ something exponentially small :-) }$
%\item There is a constant $s_{\mathit{total}}\leq ...$ such that for almost all runs of~$\Net_{\Theta}$
% the average total size of queues is equal to $s_{\mathit{total}}$.
%\end{itemize}
with stationary distribution~$\pi$ such that there exists an exponential moment of the total number of waiting jobs,
 i.e., there is $\delta > 0$ such that
 $\sum_{\vx \in \N^n} \exp(\delta \norm{\vx}) \pi(\vx)$ exists.
\end{theorem}

%\stefan{Tomas, please add the reference.}
To prove Theorem~\ref{thm:main-charact} we generalize the concept of {\em traffic equations} (see~e.g.~\cite{book:Chen01}) from the theory of Jackson networks.
Intuitively, the traffic equations express the fact that the inflow of jobs to a given queue must be equal to the outflow.
 % and that the inflow must be (strictly) smaller than the rate of the queue in order to obtain a stable system.
%\tomas{An example of a Jackson network with the corresponding traffic equations(?)}
%One of the most important properties of the traffic equations is that they
Remarkably, the traffic equations
 characterize the stability of the Jackson network.
More precisely, a Jackson network is stable %(i.e., the stochastic process induced by the network is ergodic)
 if and only if
 there is a solution of the traffic equations whose components are strictly smaller than the rates of the corresponding queues (we call such a solution {\em deficient}).
%\stefan{I would omit the last sentence of this paragraph: the paragraph is a bit repetitive.}
%In other words, a Jackson network is stable if and only if the inflow (and hence also the outflow) to every queue is strictly smaller than the rate of the queue.

We show how to extend the traffic equations so that they characterize the stability of controlled branching networks.
For a smooth presentation, we start with {\em purely stochastic} branching networks and add control later on.
Hence, the overall plan of the proof of Theorem~\ref{thm:main-charact} is as follows:
Set up traffic equations for purely stochastic branching networks and show that if there is a deficient solution of these equations, then the network is stable.
This result, presented in Section~\ref{sec:uncontr} (Proposition~\ref{prop:traffic-erg}), is of independent interest and
 requires the construction of a suitable \emph{Lyapunov function}.
Then, in Section~\ref{sec:contr}, we generalize the traffic equations to controlled branching networks and show that any ergodic scheduler determines a deficient solution
 (Proposition~\ref{prop:traffic-LP}).
This solution naturally induces a static scheduler, which, when fixed, determines an purely stochastic network with deficiently solvable traffic equations.
Propositions \ref{prop:traffic-erg} and~\ref{prop:traffic-LP} imply Theorem~\ref{thm:main-charact} and provide some additional results.

\subsection{Purely stochastic branching networks} \label{sec:uncontr}
%\tomas{Here we may move the part of prelims describing special notation for purely stochastic queues (?)}
Assume that $\Net$ is purely stochastic, i.e., there is a single action for each queue.
In such a case the CTMDP induced by the network is in fact a CTMC.
We associate the following quantities to a network, which will turn out to be crucial for its performance.
Let $\vmu := (\mu_1, \ldots, \mu_n)$.
Let $\valpha \in \Rp^n$ be the vector with $\valpha_i := \mu_0 \sum_{\vr \in R_0} \Prob_0(\vr) \vr_i$;
 i.e., $\valpha_i$ indicates the expected number of external arrivals at queue~$i$ per time unit.
Note that $\valpha \ne \vzero$, as we assume a nonzero arrival stream.
Let $A \in \Rp^{n \times n}$ be the matrix with $A_{i j} := \sum_{\vr \in R_i} \Prob_i(\vr) \vr_j$;
 i.e., $A_{i j}$ indicates the expected production of $j$-jobs when queue~$i$ fires.
W.l.o.g.\ we assume that all queues are ``reachable'', i.e., for all queues~$i$ there is $j \in \N$ with $(\valpha A^j)_i \ne 0$.
%Let $\vmu := (\mu_1, \ldots, \mu_n)$.
%Let $\valpha \in \Rp^n$ be the vector with $\valpha_i := \mu_0 \sum_{\vr \in R_0} \Prob_0(\vr) \vr_i$;
% i.e., $\valpha_i$ indicates the expected number of external arrivals at queue~$i$ per time unit.
%Let $A \in \Rp^{n \times n}$ be the matrix with $A_{i j} := \sum_{\vr \in R_i} \Prob_i(\vr) \vr_j$;
% i.e., $A_{i j}$ indicates the expected production of $j$-jobs when queue~$i$ fires.
 We define a set of \emph{traffic equations}
%\begin{equation}\label{eq:traffic_unc}
%\lambda_{j} \quad = \quad \alpha_j+\sum_{i=1}^n\lambda_{i}\cdot A_{i j}
%\quad < \quad \mu_j \quad , \quad j \in \{1, \ldots, n\},
%\end{equation}
\begin{equation}\label{eq:traffic_unc}
\lambda_{j} \quad = \quad \alpha_j+\sum_{i=1}^n\lambda_{i}\cdot A_{i j}
 \quad , \quad j \in \{1, \ldots, n\},
\end{equation}
in matrix form:
% \begin{equation}
%  \vlambda = \valpha + \vlambda A < \vmu\,. \label{eq:traffic}
% \end{equation}
\begin{equation}
  \vlambda = \valpha + \vlambda A\,. \label{eq:traffic}
 \end{equation}

We prove the following proposition.
%\newcommand{\stmtproptrafficerg}{
%Assume that %$\left(\bar{\lambda}_i\mid i\in \{1,\ldots,n\}\right)$
% $\vlambda \in \Rp^n$
%solves the traffic equations~\eqref{eq:traffic}.
%Then the following conclusions hold:
%\begin{enumerate}
% \item
%  The process~$\Net$ is ergodic, i.e., the expected return time to~$\vzero$ is finite.
% \item
%  There exists a stationary distribution~$\pi$ such that there exists an exponential moment of the total queue size,
%   i.e., there is $\delta > 0$ such that $\sum_{\vx \in \N^n} \exp(\delta \norm{\vx}) \pi(\vx)$ exists.
%% \item
%%  Define for $i \in \{1, \ldots, n\}$ the \emph{utilization of queue~$i$} by $\rho_i := \lim_{t \to \infty} \Pr(\vx_i(t) \ne 0)$.
%%  The limit~$\rho_i$ exists and we have $\rho_i = \vlambda_i / \mu_i$.
%\end{enumerate}
%}
%\begin{proposition}\label{prop:traffic-erg}
% \stmtproptrafficerg
%\end{proposition}
\newcommand{\stmtproptrafficerg}{
Assume that %$\left(\bar{\lambda}_i\mid i\in \{1,\ldots,n\}\right)$
 $\vlambda \in \Rp^n$
solves the traffic equations~\eqref{eq:traffic} and satisfies $\vlambda<\vmu$.
Then the following conclusions hold:
\begin{enumerate}
 \item
  The process~$\Net$ is ergodic, i.e., the expected return time to~$\vzero$ is finite.
 \item
  There exists a stationary distribution~$\pi$ such that there exists an exponential moment of the total queue size,
   i.e., there is $\delta > 0$ such that $\sum_{\vx \in \N^n} \exp(\delta \norm{\vx}) \pi(\vx)$ exists.
% \item
%  Define for $i \in \{1, \ldots, n\}$ the \emph{utilization of queue~$i$} by $\rho_i := \lim_{t \to \infty} \Pr(\vx_i(t) \ne 0)$.
%  The limit~$\rho_i$ exists and we have $\rho_i = \vlambda_i / \mu_i$.
\end{enumerate}
}
\begin{proposition}\label{prop:traffic-erg}
 \stmtproptrafficerg
\end{proposition}
%
%Theorem~\ref{thm:main-charact} follows directly from Proposition~\ref{prop:erg-traffic} and Proposition~\ref{prop:traffic-erg}. The most involved part is the proof of Proposition~\ref{prop:traffic-erg}.
%
The key step to the proof of Proposition~\ref{prop:traffic-erg} is to construct a so-called \emph{Lyapunov function}
 with respect to which the process~$\Net$ exhibits a ``negative drift''.
This is in fact a classical technique for showing the stability of queueing systems \cite{MeynTweedie93};
 the difficulty lies in finding a suitable Lyapunov function.
The ``drift'' of~$\Net$ is given by the \emph{mean velocity vector} $\vDelta(\vx) \in \Rp^n$ of~$\Net$,
 defined by $\vDelta(\vx) := \lim_{h\to 0^+} \Ex{\vx(t+h) - \vx(t) \mid \vx(t) = \vx}/h$.
The limit exists, is independent of~$t$, and we have
\begin{equation}
 \vDelta(\vx) = \valpha + \sum_{i: \vx_i \ne 0} \mu_i (-{\es{i}} + A_i) \,. \label{eq:mean-velocity-main}
\end{equation}
The following lemma is implicitly proved in~\cite[theorem~1]{DownMeyn97piecewise}.
\begin{lemma}\label{lem:thm1-piecewise}
 Suppose that a function $\tV : \Rp^n \to \Rp$ is two times continuously differentiable, $\tV(\vx) = 0$ implies $\vx = \vzero$,
  and that there is $\gamma > 0$ such that we have
  \[
   \vDelta(\vx) \big(\tV'(\vx)\big)^T \le -\gamma \qquad \text{ for all $\vx \ne \vzero$, }
  \]
  where $\tV'(\vx)$ denotes the gradient of~$\tV$ at~$\vx$.
 Then the conclusions of Proposition~\ref{prop:traffic-erg} holds.
\end{lemma}
Following~\cite{DownMeyn97piecewise}, we construct the Lyapunov function~$\tV$ in two stages:
 we first define a suitable \emph{piecewise linear} function~$V: \N^n \to \Rp$;
  then $V$ is \emph{smoothed} to obtain~$\tV$.
For the definition of~$V$ we need the following lemma.
\newcommand{\stmtleminvertible}{
 The matrix series $A^* := \sum_{i=0}^\infty A^i$ converges (``exists'') in~$\Rp^{n \times n}$ and is equal to $(I-A)^{-1}$.
}
\begin{lemma} \label{lem:invertible}
 \stmtleminvertible
\end{lemma}
%Observe that for any vector $\vv^T = A^* \vw^T$ with $\vw \in \R^n$, we have $\vv^T = A \vv^T + \vw^T$.
%
Define vectors $\qs{1}, \ldots, \qs{n} \in \Rp^n$ by setting $\qsT{i} := \asT{i} / \norm{\as{i}}$,
 where $\asT{i}$ is the $i$th column of~$A^*$.
Observe that we have $\vone \qsT{i} = 1$ for all~$i$.
Define the function $V : \Rp^{n} \to \Rp$ by $V(\vx) := \max_i \{\vx \qsT{i}\}$.
We will use the following property of~$V$:
\begin{lemma} \label{lem:good-max}
 If $\vzero \ne \vx \in \Rp^n$ and $\vx_i = 0$, then $\vx \qsT{i} < V(\vx)$.
\end{lemma}
Lemma~\ref{lem:good-max} is not obvious; \iftechrep{in the appendix}{in~\cite{BK12:stacs-report}} we use Farkas' lemma for the proof.
The following lemma describes the crucial ``negative drift'' property of~$V$:
\newcommand{\stmtlempiecewisenegdrift}{
 There is $\gamma > 0$ such that we have
  \[
   \vDelta(\vx) \big(V'(\vx)\big)^T \le -\gamma \qquad \text{ for all $\vx \ne \vzero$ }
  \]
  and all subgradient vectors $V'(\vx)$ of~$V$ at~$\vx$.
 More precisely, one can choose
  \[ \gamma := \min_i (\mu_i - \vlambda_i) / \norm{\as{i}} \,,
  \]
  where $\asT{i}$ is the $i$th column of~$A^*$.
}
\begin{lemma} \label{lem:piecewise-neg-drift}
 \stmtlempiecewisenegdrift
\end{lemma}
\begin{example}
 Consider the network with
 \[
   0 \btran{1} 1
    \qquad \begin{array}{l} 1 \btran{1/5} 2,2 \\ 1 \btran{4/5} \varepsilon \end{array}
    \qquad \begin{array}{l} 2 \btran{1/6} 1,2 \\ 2 \btran{5/6} \varepsilon \end{array}\,,
 \]
 arrival rate $\mu_0 = 7/30$, and $\vmu = (5/12, 7/20)$, so that $\mu_0 + \mu_1 + \mu_2 = 1$.
 Let us write $[0] := \valpha = \mu_0 (1,0)$, \ $[1] := \mu_1 (-\es{1} + A_1)$, \ $[2] := \mu_2 (-\es{2} + A_2)$, \
  $[01] := [0] + [1]$, \ $[02] := [0] + [2]$, \ $[012] := [0] + [1] + [2]$.
 These vectors are shown in Figure~\ref{fig:piecewise-neg-drift}~(a).
 %We have $\valpha + \vmu (-I + A) = [012] = (-1/8, -1/8)$, so we do not have to slow down the network and can take $\vmu' = \vmu$.
 The mean velocity vector $\vDelta(\vx)$ is one of the vectors $[0], [01], [02], [012]$,
  depending on which components of~$\vx$ are nonzero.
 The vector field in Figure~\ref{fig:piecewise-neg-drift}~(b) shows the corresponding vectors for several $\vx \in \N^2$.
 The connected line segments indicate points $\vx$ with the same value of
  $V(\vx) = \max \{\vx \qsT{1}, \vx \qsT{2}\} = \max \{ \frac56 \vx_1 + \frac16 \vx_2, \frac27 \vx_1 + \frac57 \vx_2 \}$
   (values $0.5, 1, 1.5, \ldots$).
 It can be seen from the figure that the drift is negative with respect to the gradient of~$V$, if $\vx \ne \vzero$.
\tikzset{>=latex',
 axis/.style={->},
 elem/.style={->,thick},
 drift/.style={->,ultra thick},
}
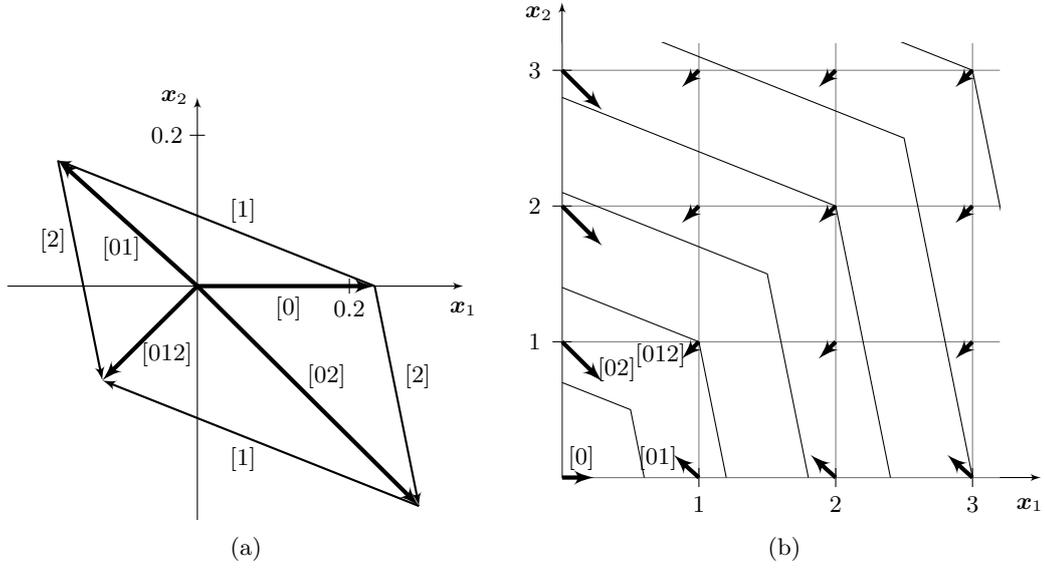
\begin{figure}
\begin{tabular}{cc}
\begin{tikzpicture}[scale=10]
 %\draw[step=0.2,gray] (-0.25,-0.35) grid (0.35,0.25);
 \draw[axis] (-0.25,0) -- (0.35,0);
 \draw[axis] (0,-0.31) -- (0,0.25);
 \draw (0.2,-0.01) -- (0.2,+0.01);
 \draw (-0.01,0.2) -- (+0.01,0.2);
 \node             at (0.2,-0.03) {$0.2$};
 \node             at (0.35,-0.03) {$\vx_1$};
 \node             at (-0.04,0.2) {$0.2$};
 \node             at (-0.03,0.25) {$\vx_2$};
 \draw[elem] (7/30,0) -- (-11/60,1/6);
 \draw[elem] (-11/60,1/6) -- (-1/8,-1/8);
 \draw[elem] (7/30,0) -- (7/24,-7/24);
 \draw[elem] (7/24,-7/24) -- (-1/8,-1/8);
 \draw[drift] (0,0) -- (7/30,0); %[0]
 \draw[drift] (0,0) -- (-11/60,1/6); %[01]
 \draw[drift] (0,0) -- (7/24,-7/24); %[02]
 \draw[drift] (0,0) -- (-1/8,-1/8); %[012]
 \node at (0.12,-0.03) {$[0]$};
 \node at (0.06,0.1) {$[1]$};
 \node at (0.06,-0.23) {$[1]$};
 \node at (0.29,-0.12) {$[2]$};
 \node at (-0.19,0.06) {$[2]$};
 \node at (-0.1,0.05) {$[01]$};
 \node at (0.17,-0.12) {$[02]$};
 \node at (-0.04,-0.09) {$[012]$};
\end{tikzpicture}
&
\begin{tikzpicture}[scale=1.8]
 \draw[axis] (0,0) -- (3.5,0);
 \draw[axis] (0,0) -- (0,3.5);
 \node at (3.42,-0.2) {$\vx_1$};
 \node at (-0.17,3.42) {$\vx_2$};
 \draw[step=1,gray] (0,0) grid (3.2,3.2);
 \draw (1,-0.07) -- (1,0.07);
 \draw (2,-0.07) -- (2,0.07);
 \draw (3,-0.07) -- (3,0.07);
 \draw (-0.07,1) -- (0.07,1);
 \draw (-0.07,2) -- (0.07,2);
 \draw (-0.07,3) -- (0.07,3);
 \node at (1,-0.2) {$1$};
 \node at (2,-0.2) {$2$};
 \node at (3,-0.2) {$3$};
 \node at (-0.2,1) {$1$};
 \node at (-0.2,2) {$2$};
 \node at (-0.2,3) {$3$};
 \draw[drift] (0,0) -- (7/30,0);
 \draw[drift] (1,0) -- (1-11/60,1/6); %[01]
 \draw[drift] (2,0) -- (2-11/60,1/6); %[01]
 \draw[drift] (3,0) -- (3-11/60,1/6); %[01]
 \draw[drift] (0,1) -- (7/24,1-7/24); %[02]
 \draw[drift] (0,2) -- (7/24,2-7/24); %[02]
 \draw[drift] (0,3) -- (7/24,3-7/24); %[02]
 \draw[drift] (1,1) -- (1-1/8,1-1/8); %[012]
 \draw[drift] (2,1) -- (2-1/8,1-1/8); %[012]
 \draw[drift] (3,1) -- (3-1/8,1-1/8); %[012]
 \draw[drift] (1,2) -- (1-1/8,2-1/8); %[012]
 \draw[drift] (2,2) -- (2-1/8,2-1/8); %[012]
 \draw[drift] (3,2) -- (3-1/8,2-1/8); %[012]
 \draw[drift] (1,3) -- (1-1/8,3-1/8); %[012]
 \draw[drift] (2,3) -- (2-1/8,3-1/8); %[012]
 \draw[drift] (3,3) -- (3-1/8,3-1/8); %[012]
 \clip (0,0) rectangle (3.2,3.2);
% \draw (0,0) circle (1);
% \draw (0,0) circle (1.414213562);
% \draw (0,0) circle (1.732050808);
% \draw (0,0) circle (2);
% \draw (0,0) circle (2.236067977);
% \draw (0,0) circle (2.449489743);
% \draw (0,0) circle (2.645751311);
% \draw (0,0) circle (2.828427124);
% \draw (0,0) circle (3.);
% \draw (0,0) circle (3.162277660);
% \draw (0,0) circle (3.316624790);
% \draw (0,0) circle (3.464101616);
% \draw (0,0) circle (3.605551275);
% \draw (0,0) circle (3.741657387);
% \draw (0,0) circle (3.872983346);
% \draw (0,0) circle (4);
% \draw (0,0) circle (4.123105626);
% \draw (0,0) circle (4.242640686);
% \draw (0,0) circle (4.358898944);
% \draw (0,0) circle (4.472135954);
% \draw (0,0) circle (4.582575695);
 \draw (1*0.6,0) -- (1*0.5,1*0.5) -- (0,1*0.7);
 \draw (2*0.6,0) -- (2*0.5,2*0.5) -- (0,2*0.7);
 \draw (3*0.6,0) -- (3*0.5,3*0.5) -- (0,3*0.7);
 \draw (4*0.6,0) -- (4*0.5,4*0.5) -- (0,4*0.7);
 \draw (5*0.6,0) -- (5*0.5,5*0.5) -- (0,5*0.7);
 \draw (6*0.6,0) -- (6*0.5,6*0.5) -- (0,6*0.7);
 \node at (0.14,0.14) {$[0]$};
 \node at (0.7,0.14) {$[01]$};
 \node at (0.4,0.8) {$[02]$};
 \node[inner sep=0mm] at (0.72,0.9) {$[012]$};
\end{tikzpicture}
\\
(a)
&
(b)
\end{tabular}
\caption{Illustration of negative drift.}
\label{fig:piecewise-neg-drift}
\end{figure}
\end{example}
\begin{proof}[Proof of Lemma~\ref{lem:piecewise-neg-drift}]
 Let $\vx \ne \vzero$.
 We need to show $\vDelta(\vx) \qsT{i} \le -\gamma$ for all $i$ with $\vx \qsT{i} = V(\vx)$.
 W.l.o.g.\ we assume that $\vx \qsT{1} = V(\vx)$ and show only $\vDelta(\vx) \qsT{1} \le -\gamma$.
 By Lemma~\ref{lem:good-max} we have $\vx_1 \ne 0$.
 It follows from the property $(I-A) A^* = I$ and the definition of~$\qs{1}$ that we have
  \begin{equation} \label{eq:q1-cancellations-main}
   (-{\es{1}} + A_1) \qsT{1} = -1/\norm{\as{1}} \qquad \text{and} \qquad (-\es{i} + A_i) \qsT{1} = 0 \quad \text{ for $2 \le i \le n$.}
  \end{equation}
 Hence we have:
 \begin{align*}
  \vDelta(\vx) \qsT{1} & = \left( \valpha + \sum_{i: \vx_i \ne 0} \mu_i (-{\es{i}} + A_i) \right) \qsT{1}  && \text{by~\eqref{eq:mean-velocity-main}} \\
                        & = \valpha \qsT{1} - \mu_1 / \norm{\as{1}} && \text{by~\eqref{eq:q1-cancellations-main} and $\vx_1 \ne 0$} \\
                        & \le -\gamma + \valpha \qsT{1} - \vlambda_1 / \norm{\as{1}} && \text{by the definition of~$\gamma$} \\
                        & = -\gamma + \left( \valpha + \sum_{i=1}^n \vlambda_i (-{\es{i}} + A_i) \right) \qsT{1}
                                  && \text{by~\eqref{eq:q1-cancellations-main}} \\
                        & = -\gamma + \left( \valpha + \vlambda (-I+A) \right) \qsT{1} \\
                        & = -\gamma + \vzero \qsT{1} = -\gamma && \text{by the traffic equation.}
 \end{align*}
\end{proof}
%By smoothing~$V$ as formally described in the appendix of~\cite{DownMeyn97piecewise},
% we obtain a two times continuously differentiable function~$\tV$ satisfying the conditions in Lemma~\ref{lem:thm1-piecewise}.
Using integration, one can obtain a two times continuously differentiable function~$\tV$ satisfying the conditions in Lemma~\ref{lem:thm1-piecewise}:
the function~$V$ is smoothed by defining $\tV(\vx)$, for all~$\vx$, as an ``average'' of the values~$V(\vy)$ 
  where $\vy$ belongs to a small ball around~$\vx$;
 see the appendix of~\cite{DownMeyn97piecewise} for the formal details.
This concludes the proof of Proposition~\ref{prop:traffic-erg}.

%Define for $i \in \{1, \ldots, n\}$ the \emph{utilization of queue~$i$} by $\rho_i := \lim_{t \to \infty} \Pr(\vx_i(t) \ne 0)$.
%  The limit~$\rho_i$ exists and we have $\rho_i = \vlambda_i / \mu_i$.
%
%Finally, let us give an intuition behind the part 3. of Proposition~\ref{prop:traffic-erg} (a complete formal proof is given in Appendix~\ref{}).
%By the Ergodic theorem (see~e.g.~\cite{Norris}[Theorem~3.8.1]), the
%utilization $\rho_i$ is almost surely equal to the proportion of time when the queue $i$ is non-empty, i.e.,
%\[
%\Pr\left(\lim_{t\rightarrow \infty} \frac{1}{t} \int_{0}^t I[\vx_i(s)\not = 0] ds\ =\ \rho_i\right)\quad = \quad 1
%\]
%Here $I[\vx_i(s)\not = 0]$ is equal to $1$ if $\vx_i(s)\not = 0$, and is equal to $0$ otherwise.
% %
%However, it is straightforward to show that this proportion is equal to the frequency of firing the queue $i$ times the average time it takes to fire the queue $i$ whenever it is non-empty.
% The average time to fire the non-empty queue $i$ is of course equal to $\frac{1}{\mu_i}$.
%We show in the appendix (Lemma~\ref{lem:ex-freq}) that the frequency of firing the queue $i$ solves the traffic equations (even for controlled queues). For purely stochastic queues, the traffic equations have, in fact, a unique solution, $\vlambda$, because the matrix $I-A$ is invertible. So the frequency of firing the queue $i$ is {\em equal} to $\lambda_i$.
%Thus the proportion of time the queue $i$ is non-empty (and thus also $\rho_i$) is equal to $\frac{\lambda_i}{\mu_i}$.
%
%
%\stefan{Tomas, please add the proof of item 3. of Prop.~\ref{prop:traffic-erg}.}

\subsection{Controlled branching networks} \label{sec:contr}
In this subsection we generalize the traffic equations~(\ref{eq:traffic}) to deal with an arbitrary controlled branching network $\mathcal{N}$.
%
%capturing existence
%, show that any solution induces
%a static randomized ergodic scheduler, and finally,
%show that the traffic equations have a solution if $\mathcal{N}$ admits an ergodic scheduler.
To obtain a distribution on actions for a static randomized ergodic scheduler,
we assign variables to actions instead of queues, i.e., for every action $\xi$ of the network we introduce a variable $\lambda_{\xi}$ capturing the {\em rate of firing the action $\xi$}.
Denote by $\bSigma$ the set $\bigcup_{i=1}^n \Sigma_i$.
Given $\zeta\in \bSigma$ and $j\in \{1,\ldots,n\}$, we denote by $A_{\zeta j}$ the average number of jobs added to the queue $j$ when the action $\zeta$ fires,
 i.e., for $\zeta\in \Sigma_i$ we set
\[
A_{\zeta j} := \sum_{\vr\in R_i(\zeta)} \Prob_i(\zeta)(\vr) \cdot \vr_j\,.
\]
%Let $\valpha \in \Rp^n$ be the vector with
%\[
%\valpha_i := \mu_0 \sum_{\vr \in R_0} \Prob_0(\vr) \vr_i
%\]
%i.e., $\valpha_i$ indicates the expected number of external arrivals at queue~$i$ per time unit.
 %
We generalize~\eqref{eq:traffic} to the \emph{traffic LP} presented in Figure~\ref{fig:LP-delta},
 where the variable~$\delta$ is intended to bound, for all~$j$, the probability that queue~$j$ is busy.
\begin{figure}[ht]
\begin{center}
\[
\min\, \delta \ \text{subject to}
\]
\begin{align*}
\sum_{\xi\in \Sigma_j} \lambda_{\xi} \quad & = \quad \valpha_j+\sum_{i=1}^n \sum_{\zeta\in \Sigma_i} \lambda_{\zeta}\cdot A_{\zeta j} & j\in \{1,\ldots,n\} \\
\delta \quad & \geq \quad \frac{\sum_{\xi\in \Sigma_j} \lambda_{\xi}}{\mu_j} & \quad  \quad j\in \{1,\ldots,n\} \\
\lambda_{\xi} \quad & \geq \quad 0 & \xi\in \bar{\Sigma}
\end{align*}

\end{center}
\caption{The traffic LP.}
\label{fig:LP-delta}
\end{figure}

We prove the following
%
%For every $i\in \{1,\ldots,n\}$ we put
%\stefan{If we assume a fixed rate~$\mu_i$ for each queue, we need to change the preliminaries accordingly (and justify this assumption quickly).}
%\begin{equation}\label{eq:c-traffic}
%\sum_{\xi\in \Sigma_i} \lambda_{\xi} \quad = \quad \valpha_i+\sum_{j=1}^n \sum_{\zeta\in \Sigma_j} \lambda_{\zeta}\cdot A_{\zeta i}
%\quad < \quad \mu_i
%\end{equation}
%Here each $\lambda_\xi$ is a variable corresponding to the {\em traffic through the action $\xi$}.
%\stefan{The following sentence seems strange. Why unique and why is the sum = 1? Does that use rescaling time?}
%
%\begin{proposition}\
% %One can compute in polynomial time a linear program LP such that the following holds:
% \begin{enumerate}
%   \item
%    Given any non-negative solution of~(\ref{eq:c-traffic}), one can compute in polynomial time a static randomized ergodic scheduler $\Theta_{static}$ for~$\Net$.
%    This scheduler $\Theta_{static}$ minimizes $\max_i \rho_i$ among all static randomized schedulers (or even among all schedulers?).
%   \item
%    If there exists an arbitrary ergodic scheduler for~$\Net$, then~(\ref{eq:c-traffic}) has a solution.
% \end{enumerate}
%\end{proposition}
\begin{proposition}\label{prop:traffic-LP}\
\begin{enumerate}
\item
    If there exists an arbitrary ergodic scheduler for~$\Net$, then the traffic LP can be solved with $\min \delta<1$.
   \item If the traffic LP is solved with $\min \delta<1$, one can compute in polynomial time a static randomized ergodic scheduler $\Theta_s$ for~$\Net$. Moreover, denoting by $\rho_i$ the {\em utilization} $\lim_{t \to \infty} \Pr(\vx_i(t) \ne 0)$
   of the queue $i$, the scheduler $\Theta_s$ minimizes %the {\em maximal utilization}
     $\max_i \rho_i$ among all memoryless ergodic schedulers.
\end{enumerate}
%Morevoer, define for $i \in \{1, \ldots, n\}$ the \emph{utilization of queue~$i$} by $\rho_i := \lim_{t \to \infty} \Pr(\vx_i(t) \ne 0)$.
%  The limit~$\rho_i$ exists and we have $\rho_i = \vlambda_i / \mu_i$.
%
%In addition, the scheduler $\Theta_s$ minimizes the maximal utilization %$\max_i \rho_i$
%among all memoryless ergodic schedulers. % (or even among all schedulers?).
Hence one can decide in polynomial time whether an arbitrary ergodic scheduler exists;
 if yes, one can compute in polynomial time a static randomized ergodic scheduler.
\end{proposition}
Let us first concentrate on part~1. Let $\Theta$ be an ergodic scheduler.
Roughly speaking, we prove that a feasible solution of the traffic LP can be constructed using (limit) frequencies of firing individual actions in $\mathcal{N}_{\Theta}$.
Formally, given a run $\omega$ of $\mathcal{N}_{\Theta}$, $t\in \Rp$ and $\xi\in \bSigma$, we denote by $O^{\leq t}_{\xi}(\omega)$ the number of times the action $\xi$ is fired up to time $t$ on $\omega$. %First, let us consider memoryless $\Theta$.
For memoryless $\Theta$ we have the following result.
%
%\begin{lemma}\label{prop:erg-traffic}
%Assume that $\Theta$ is a {\em memoryless} ergodic scheduler.
%For every $\xi\in \bSigma$ there is a constant $O_{\xi}$ such that for almost all runs $\omega$ of $\mathcal{N}_{\Theta}$ the limit \[\lim_{t\rightarrow{} \infty} \frac{O^{\leq t}_{\xi}(\omega)}{t}\]
%exists and is equal to $O_{\xi}$. There is $\bar{\delta}<1$ such that $\left(\bar{\delta},O_{\xi}\mid \xi\in \bSigma\right)$ solves the traffic LP.
%Moreover, for every $i\in \{1,\ldots,n\}$ the utilization $\rho_i$ of the queue $i$ in $\Net_{\Theta}$ is equal to $\frac{\sum_{\xi\in \Sigma_i} O_{\xi}}{\mu_i}$.
%\end{lemma}
\newcommand{\stmtpropergtraffic}{
Assume that $\Theta$ is a {\em memoryless} ergodic scheduler.
For every $\xi\in \bSigma$ there is a constant $O_{\xi}$ such that for almost all runs $\omega$ of $\mathcal{N}_{\Theta}$ the limit \[\lim_{t\rightarrow{} \infty} \frac{O^{\leq t}_{\xi}(\omega)}{t}\]
exists and is equal to $O_{\xi}$. There is $\bar{\delta}<1$ such that $\left(\bar{\delta},O_{\xi}\mid \xi\in \bSigma\right)$ solves the traffic LP.
Moreover, for every $i\in \{1,\ldots,n\}$ the utilization $\rho_i$ of the queue $i$ in $\Net_{\Theta}$ is equal to $\frac{\sum_{\xi\in \Sigma_i} O_{\xi}}{\mu_i}$.
}
\begin{lemma}\label{prop:erg-traffic}
\stmtpropergtraffic
\end{lemma}
We prove Lemma~\ref{prop:erg-traffic} \iftechrep{in Appendix~\ref{app:contr}}{in~\cite{BK12:stacs-report}}.
If there exists an arbitrary (i.e.\ possibly history-dependent) ergodic scheduler,
 then by~Theorem~7.3.8 of \cite{book:Puterman} there exists also a memoryless (and deterministic) ergodic scheduler.%
\footnote{To be formally correct, we apply Theorem~7.3.8 of \cite{book:Puterman} to the {\em embedded} discrete time MDP and obtain a scheduler which returns to the state $\vzero$ in finitely many steps (on average). As there are only finitely many rates in our system, this means that also the expected return time to $\vzero$ is finite.}
This fact, combined with Lemma~\ref{prop:erg-traffic}, implies part~1.\ of Proposition~\ref{prop:traffic-LP}.

Now let us concentrate on part 2.\ of Proposition~\ref{prop:traffic-LP}.
\begin{lemma}\label{lem:LP-sched}
Any feasible solution $\left(\bar{\delta},\bar{\lambda}_{\xi}\mid \xi\in \bSigma\right)$ of the traffic LP with $\bar{\delta}<1$ {\em induces} a static randomized ergodic scheduler whose utilization of any queue $i$ is equal to $\frac{\sum_{\xi\in \Sigma_i} \bar{\lambda}_{\xi}}{\mu_i}$.
\end{lemma}
\begin{proof}
We construct a static randomized scheduler $\Theta$ which chooses an action $\xi\in \Sigma_i$ for the queue $i$ with probability
\begin{equation}\label{eq:prob_act}
P_{\xi}=\frac{\bar{\lambda}_{\xi}}{\sum_{\zeta\in \Sigma_i} \bar{\lambda}_{\zeta}}\quad , \quad \sum_{\zeta\in \Sigma_i} \bar{\lambda}_{\zeta}>0\,.
\end{equation}
Otherwise, if $\sum_{\zeta\in \Sigma_i} \bar{\lambda}_{\zeta}=0$, we may control the queue $i$ arbitrarily because no jobs ever come to the queue.  We further assume (w.l.o.g.) that such queues have been removed from the network, i.e., that $P_{\xi}$ is defined using~(\ref{eq:prob_act}) for all $\xi\in \bar{\Sigma}$.
Note that $\sum_{\xi\in \Sigma_i} P_{\xi}=1$ for every $i\in \{1,\ldots,n\}$.

Fixing the scheduler $\Theta$ we obtain a purely stochastic branching network whose traffic equations are deficiently solvable.
Formally, we define a new purely stochastic branching network $\mathcal{N}'$ with $n$ queues with
 the same arrival rate, the same arrival production function and the same queue rates as~$\mathcal{N}$.
Further, $\Net'$ has
 $R'_i=\bigcup_{\xi\in \Sigma_i} R_i(\xi)$ and the following production functions $\Prob'_i$ associated to queues:
\[
\Prob'_i(\vr)=\sum_{\xi\in \Sigma_i} P_{\xi}\cdot \Prob_i(\xi)(\vr)\quad , \quad\vr\in R'_i
\]
(Here we formally assume $\Prob_i(\xi)(\vr)=0$ for $\vr\not\in R_{\xi}$.)
The traffic equations~\eqref{eq:traffic_unc} for $\Net'$ have the following form:
%\begin{equation}\label{eq:traffic-derived}
%\lambda_{j} \quad = \quad \alpha_j+\sum_{i=1}^n\lambda_{i}\cdot A'_{i j}
%\quad < \quad \mu_j \quad , \quad j \in \{1, \ldots, n\}
%\end{equation}
\begin{equation}\label{eq:traffic-derived}
\lambda_{j} \quad = \quad \alpha_j+\sum_{i=1}^n\lambda_{i}\cdot A'_{i j}
\quad , \quad j \in \{1, \ldots, n\}
\end{equation}
with
\begin{align*}
A'_{i j} := \sum_{\vr\in R'_i} & \Prob'_i(\vr)\cdot \vr_j=
\sum_{\xi\in \Sigma_i}P_{\xi}  \sum_{\vr\in R'_i} \Prob_i(\xi)(\vr)\cdot \vr_j= \\
& =\sum_{\xi\in \Sigma_i} \frac{\bar{\lambda}_{\xi}}{\sum_{\zeta\in \Sigma_i} \bar{\lambda}_{\zeta}}\sum_{\vr\in R'_i} \Prob_i(\xi)(\vr)\cdot \vr_j =
\sum_{\xi\in \Sigma_i} \frac{\bar{\lambda}_{\xi}}{\sum_{\zeta\in \Sigma_i} \bar{\lambda}_{\zeta}}\cdot A_{\xi j}\,.
\end{align*}
Setting $\lambda_i := \sum_{\xi\in \Sigma_i} \bar{\lambda}_{\xi}$ for every $i\in \{1,\ldots,n\}$,
 we obtain $\lambda_i A'_{i j} = \sum_{\xi\in \Sigma_i} \bar{\lambda}_{\xi} A_{\xi j}$.
If we put this equality into the first equation of the traffic LP, we see that $(\lambda_1, \ldots, \lambda_n)$ solves~\eqref{eq:traffic-derived}. Also, $\lambda_j<\mu_j$ for all $j \in \{1, \ldots, n\}$.
%%%Denoting
%%%%\stefan{Define $A'$ in terms of $\Prob'$ so that Sec. 3.1 can be applied.}
%%%\begin{align*}
%%%A'_{i j} :=  \sum_{\vr\in R'_i} & \Prob'_i(\vr)\cdot \vr_j=
%%%\sum_{\xi\in \Sigma_i}P_{\xi}  \sum_{\vr\in R'_i} \Prob_i(\xi)(\vr)\cdot \vr_j= \\
%%%& =\sum_{\xi\in \Sigma_i} \frac{\bar{\lambda}_{\xi}}{\sum_{\zeta\in \Sigma_i} \bar{\lambda}_{\zeta}}\sum_{\vr\in R'_i} \Prob_i(\xi)(\vr)\cdot \vr_j =
%%%\sum_{\xi\in \Sigma_i} \frac{\bar{\lambda}_{\xi}}{\sum_{\zeta\in \Sigma_i} \bar{\lambda}_{\zeta}}\cdot A_{\xi j}
%%%\end{align*}
%%%the traffic equations for $\mathcal{N}'$ have the following form:
%%%\begin{equation}\label{eq:traffic-derived}
%%%\lambda_{j} \quad = \quad \alpha_j+\sum_{i=1}^n\lambda_{i}\cdot A'_{i j}
%%%\quad < \quad \mu_j \quad , \quad j \in \{1, \ldots, n\},
%%%\end{equation}
%%%%in matrix form:
%%%% \begin{equation*}
%%%%  \vlambda = \valpha + \vlambda A' < \vmu\,. \label{eq:traffic-derived-matrix}
%%%% \end{equation*}
%%%Clearly, (\ref{eq:traffic-derived}) is solved by assigning $\sum_{\zeta\in \Sigma_i} \bar{\lambda}_{\zeta}$
%%%to $\lambda_{i}$ for every $i\in \{1,\ldots,n\}$.
%Now Theorem~\ref{thm:main-charact-alt} follows from
Proposition~\ref{prop:traffic-erg} then implies that the scheduler $\Theta$ is ergodic.

Finally, let us concentrate on the utilization. Note that the utilization of any queue $i$ is the same in $\Net'$ as in
$\Net_{\Theta}$, so it suffices to concentrate on $\Net'$.
Observe that the matrix $I-A'$ is invertible by Lemma~\ref{lem:invertible}. %as argued in Section~\ref{sec:uncontr}.
This means that
$(\lambda_1, \ldots, \lambda_n)$ is, in fact, the {\em unique} solution of~\eqref{eq:traffic-derived}.
Then however, by Lemma~\ref{prop:erg-traffic}, the utilization $\rho_i$ of queue~$i$ in $\Net'$ (and thus also in $\Net_{\Theta}$) is equal to
$\frac{\lambda_i}{\mu_i}=\frac{\sum_{\xi\in \Sigma_i}\bar{\lambda}_{\xi}}{\mu_i}$.
\end{proof}
%
%. Moreover, it is easy to observe that $\frac{\sum_{\zeta\in \Sigma_i} \bar{\lambda}_{\zeta}}{\mu_i}$ is equal to the utilization of the queue $i$ in $\Net_{\Theta}$ as, intuitively, $\sum_{\zeta\in \Sigma_i} \bar{\lambda}_{\zeta}$ corresponds to the rate with which the jobs are coming to the queue $i$ and $\frac{1}{\mu_i}$ corresponds to the average time to process one job in the queue $i$. (We provide a detailed argument in Appendix~\ref{app:contr}.)
%Apparently minimizing the $\delta$ minimizes the utilization over all strategies induced by solutions of the traffic LP. Hence, it suffices to define $\Theta_s$ to be a scheduler induced by a solution of the traffic LP.
%This finishes a proof of Proposition~\ref{prop:traffic-LP}.
%
%Moreover, by Proposition~\ref{},
%\[
%\max_{i} \rho_i\quad =\quad \max_i \frac{\sum_{\zeta\in \Sigma_i} \bar{\lambda}_{\zeta}}{\mu_i}\quad \geq\quad \bar{\delta}
%\]
%
%Let us fix an arbitrary ergodic scheduler $\Theta$. By~Theorem~7.3.8 of \cite{book:Puterman}, we may safely assume that $\Theta$ is memoryless and deterministic.
%
To complete the proof of Proposition~\ref{prop:traffic-LP}, we consider the problem of minimizing the maximal utilization $\max_i \rho_i$.
Let $\Theta_s$ be a static randomized ergodic scheduler induced by a solution of the traffic LP in the sense of Lemma~\ref{lem:LP-sched}
 (here we consider a solution which minimizes $\delta$).
Observe that the scheduler $\Theta_s$ minimizes $\max_i \rho_i$ among all schedulers induced by solutions of the traffic LP.
However, by Lemmas \ref{prop:erg-traffic}~and~\ref{lem:LP-sched}, for every memoryless scheduler $\Theta$
 there exists a static randomized scheduler induced by a solution of the traffic LP which has the same utilization of each queue as $\Theta$.
Thus $\Theta_s$ minimizes $\max_i \rho_i$ among all memoryless ergodic schedulers.

\section{Conclusions} \label{sec:conclusions}
We have suggested and studied controlled branching networks,
 a queueing model which extends Jackson networks
  by nondeterministic and branching features as required to model parallel systems.
Although much of the classical theory (such as product-form stationary distributions) no longer holds for controlled branching networks,
 we have shown that the traffic equations can be generalized.
This enabled us to construct a suitable Lyapunov function
 which we have used to establish strong stability properties.
We have shown for the controlled model that static randomized schedulers are sufficient to achieve those strong stability properties.
Linear programming can be used to efficiently compute such a scheduler, which at the same time minimizes the maximal queue utilization.

Future work should include the investigation of more performance measures, e.g., the long-time average queue size.
Can non-static schedulers help to minimize it?

%\input{old-note}
%---------------------
\bibliographystyle{plain} %oder alpha oder splncs
\bibliography{db}
%---------------------

\iftechrep{
\newpage
\appendix
\section{Proofs}

\subsection{Proof of Lemma~\ref{lem:invertible}}

\begin{qlemma}{\ref{lem:invertible}}
 \stmtleminvertible
\end{qlemma}

\begin{proof}
By the traffic equation $\vlambda = \valpha + \vlambda A$ we have $\vlambda A \le \vlambda$, with $\vlambda > \vzero$.
By~\cite[Theorem~2.1.11]{book:BermanP} this implies that the spectral radius of~$A$ is at most~$1$
 (where the spectral radius is the largest absolute value of the eigenvalues of~$A$).

To show the statement of the lemma, it now suffices to show that $I-A$ is invertible.
Consider the (monotone) function $\vf : \Rp^n \to \Rp^n$ with $\vf(\vx) := \valpha + \vx A$.
By the traffic equation, $\vlambda$ is a fixed point of~$\vf$.
Assume for a contradiction that $I-A$ is singular, i.e., there is $\vy \in \R^n$ with $\vy \ne \vzero$ and $\vy A = \vy$.
Then $\vlambda + r \vy$ is a fixed point for all $r \in \R$.
Choose $r$ so that $\vz := \vlambda + r \vy \in \Rp^n$, but $\vz_i = 0$ for some $i \in \{1, \ldots, n\}$.
By the monotonicity of~$\vf$, all points $\vu \in \{ \vzero, \vf(\vzero), \vf(\vf(\vzero)), \ldots\}$
 satisfy $\vzero \le \vu \le \vz = \vf(\vz)$.
It follows that $\vu_i = 0$ for all~$\vu$.
This contradicts the fact that queue~$i$ is ``reachable''
 (recall that ``queue~$i$ is reachable'' means there is $j \in \N$ with $(\valpha A^j)_i \ne 0$).
\end{proof}

\subsection{Proof of Proposition~\ref{prop:traffic-erg}}
In this section we complete the proof of

\begin{qproposition}{\ref{prop:traffic-erg}}
 \stmtproptrafficerg
\end{qproposition}

Recall from Lemma~\ref{lem:invertible} that the matrix series $A^* := \sum_{i=0}^\infty A^i$ converges (``exists'') in~$\Rp^{n \times n}$ and equals $(I-A)^{-1}$.
Observe that for any vector $\vv^T = A^* \vw^T$ with $\vw \in \R^n$, we have $\vv^T = A \vv^T + \vw^T$.
We prove Lemma~\ref{lem:good-max}:

\begin{qlemma}{\ref{lem:good-max}}
 If $\vzero \ne \vx \in \Rp^n$ and $\vx_i = 0$, then $\vx \qsT{i} < V(\vx)$.
\end{qlemma}
\begin{proof}
For the proof we use the following notation:
For a matrix~$M$ we denote by $M_{i..j}$ the matrix obtained by restricting~$M$ to its rows indexed by $i, \ldots, j$.
Similarly, we denote by $M_{i..j,k..\ell}$ the matrix~$M$ restricted to the corresponding rows and columns.

Let $Q \in \Rp^{n \times n}$ denote the matrix whose columns are $\qsT{1}, \ldots, \qsT{n}$.
Note that $\vone Q = \vone$.
We have $Q = A^* D$ where $D \in \Rp^{n \times n}$ is the diagonal matrix with $1/D_{ii} = ( \vone A^* )_i$.
Further, we define the matrix $P \in \{-1,0,1\}^{n \times (n-1)}$ with
 \[
  P := \begin{pmatrix}
        -1 \ & \ -1 \ & \ -1 \ & \cdots & \ -1 \ & \ -1 \\
        1  &  0 &  0 & \cdots & 0  &  0 \\
        0  &  1 &  0 & \cdots & 0  &  0 \\
           &    &    & \ddots  \\
        0  & 0  &  0 & \cdots & 0  & 1
       \end{pmatrix} \;.
 \]

W.l.o.g.\ let $\vx_1 = 0$.
With the notation above and writing $\bar\vx := (x_2, \ldots, x_n) \in \Rp^{n-1}$,
 the lemma states that
\begin{equation} \label{eq:max-inequality}
 \text{there is no } \bar\vx \in \Rp^{n-1} \text{ with } \bar\vx \ne \vzero \text{ and } \bar\vx Q_{2..n} P \le \vzero \,.
\end{equation}
We use Farkas' lemma for the proof:
\begin{lemma}[Farkas' lemma]
 Let $M \in \R^{m \times n}$ and $\vb \in \R^m$.
 Then exactly one of the following is true:
 \begin{itemize}
   \item[1.]
    There exists $\vx \in \R^m$ with $\vx M \le \vzero$ and $\vx \vb^T > 0$.
   \item[2.]
    There exists $\vy \in \R^n$ with $M \vy^T = \vb^T$ and $\vy \ge \vzero$.
 \end{itemize}
\end{lemma}
Using Farkas' lemma we prove~\eqref{eq:max-inequality} by exhibiting $\vb > \vzero$ and $\vy \ge \vzero$ such that $Q_{2..n} P \vy^T = \vb^T$.
We choose
\[
 \vb^T := (A_{2..n,2..n})^* \vone^T \ge \vone^T > \vzero^T \,.
\]
We have $\vone Q = \vone$, which implies $\vone Q^{-1} = \vone Q Q^{-1} = \vone$,
 hence $P (Q^{-1})_{2..n} = Q^{-1} - E$ where $E \in \{0,1\}^{n \times n}$ denotes the matrix with $\vone$ in the first row and $\vzero$ in the other rows.
It follows
\begin{equation} \label{eq:PQP=QP}
 P (Q^{-1})_{2..n} P = Q^{-1} P \,.
\end{equation}
Setting
 \[
  \vy^T := (Q^{-1})_{2..n} P \vb^T
 \]
we therefore have
 \[
  Q_{2..n} P \vy^T = Q_{2..n} P (Q^{-1})_{2..n} P \vb^T \mathop{=}^{\eqref{eq:PQP=QP}} Q_{2..n}  Q^{-1} P \vb^T = P_{2..n} \vb^T = \vb^T \,,
 \]
 as desired.
It remains to prove $\vy \ge \vzero$.
As $Q = A^* D$, we have $Q^{-1} = D^{-1} (I-A)$ with $D^{-1}$ nonnegative, so it suffices to prove that $(I-A)_{2..n} P \vb^T \ge \vzero^T$.
Indeed we have
 \[
  (I-A)_{2..n} P \vb^T \ = \ (I - A_{2..n} P) \vb^T \ = \ \vb^T - A_{2..n} P \vb^T \ \ge \ \vb^T - A_{2..n,2..n} \vb^T \ = \ \vone^T \,,
 \]
 where the inequality holds as $\vb^T$ and the first column of~$A_{2..n}$ are nonnegative and the first row of~$P$ is negative.

\end{proof}

\noindent Recall the following lemma:

\begin{qlemma}{\ref{lem:piecewise-neg-drift}}
 \stmtlempiecewisenegdrift
\end{qlemma}
%\begin{proof}
% Let $\vx \ne \vzero$.
% We need to show $\vDelta(\vx) \qsT{i} \le -\gamma$ for all $i$ with $\vx \qsT{i} = V(\vx)$.
% W.l.o.g.\ we assume that $\vx \qsT{1} = V(\vx)$ and show only $\vDelta(\vx) \qsT{1} \le -\gamma$.
% By Lemma~\ref{lem:good-max} we have $\vx_1 \ne 0$.
% It follows from the property $(I-A) A^* = I$ and the definition of~$\qs{1}$ that we have
%  \begin{equation} \label{eq:q1-cancellations}
%   (-{\es{1}} + A_1) \qsT{1} = -1/\norm{\as{1}} \qquad \text{and} \qquad (-\es{i} + A_i) \qsT{1} = 0 \quad \text{ for $2 \le i \le n$.}
%  \end{equation}
% Hence we have:
% \begin{align*}
%  \vDelta(\vx) \qsT{1} & = \left( \valpha + \sum_{i: \vx_i \ne 0} \mu_i (-{\es{i}} + A_i) \right) \qsT{1}  && \text{by~\eqref{eq:mean-velocity}} \\
%                        & = \valpha \qsT{1} - \mu_1 / \norm{\as{1}} && \text{by~\eqref{eq:q1-cancellations} and $\vx_1 \ne 0$} \\
%                        & \le -\gamma + \valpha \qsT{1} - \vlambda_1 / \norm{\as{1}} && \text{by the definition of~$\gamma$} \\
%                        & = -\gamma + \left( \valpha + \sum_{i=1}^n \vlambda_i (-{\es{i}} + A_i) \right) \qsT{1}
%                                  && \text{by~\eqref{eq:q1-cancellations}} \\
%                        & = -\gamma + \left( \valpha + \vlambda (-I+A) \right) \qsT{1} \\
%                        & = -\gamma + \vzero \qsT{1} = -\gamma && \text{by the traffic equation.}
% \end{align*}
%\qed
%\end{proof}

\medskip
Proposition~\ref{prop:traffic-erg} follows from Lemma~\ref{lem:piecewise-neg-drift}
 using exactly the reasoning from theorem~1 and lemma~5 of~\cite{DownMeyn97piecewise}.

%\input{sec-drift}
%\section{Proofs of Section~\ref{sec:contr}}
\subsection{Proof of Lemma~\ref{prop:erg-traffic}} \label{app:contr}
%Let us fix a network $\mathcal{N}$ and
%a memory-less deterministic scheduler $\Theta$ such that
%
%Note that since $\Theta$ is memory-less and ergodic,
% almost all runs of the process $\mathcal{N}_{\Theta}$ reach $\vzero$ infinitely many times and the mean time between each two visits to $\vzero$ is finite (by the strong law of large numbers).
%
%Given a run $\omega$ of $\mathcal{N}_{\Theta}$, $t\in \Rp$ and $\xi\in \bigcup_{i=1}^n \Sigma_i$, we denote by $O^{\leq t}_{\xi}(\omega)$ the number of times the action $\xi$ fires up to time $t$ on $\omega$.

\begin{qlemma}{\ref{prop:erg-traffic}}
\stmtpropergtraffic
\end{qlemma}
\medskip

We divide our proof of Lemma~\ref{prop:erg-traffic} into four claims.
\begin{claim}\label{lem:ex-freq}
For every $\xi\in \bar{\Sigma}$ there is a constant $O_{\xi}$ such that for almost all runs $\omega$ of $\Net_{\Theta}$ the limit
\[
\lim_{t\rightarrow{} \infty} \frac{O^{\leq t}_{\xi}(\omega)}{t}
\]
exists and is equal to $O_{\xi}$.
\end{claim}
\begin{proof}
Given an action $\xi\in \bar{\Sigma}$ and $i\geq 1$, denote by $M^{\xi}_i$ a random variable giving the number of times the action $\xi$ is fired between the $i$-th and $i+1$-st visit to $\vzero$ (note that the network starts with all empty queues so the first visit happens at time $0$). Observe that all $M^{\xi}_i$
are identically distributed.
Also, as the expected return time to $\vzero$ is finite and there are only finitely many distinct rates (i.e. the expected time between two state transitions is bounded from below by a positive constant), the expectation $\Ex{M^{\xi}_1}$ is finite. By the strong law of large numbers (see~e.g.~\cite{Billingsley:book}), there is a constant $M^{\xi}$ such that for almost all runs $\omega$ the limit $\lim_{n\rightarrow \infty} \frac{\sum_{i=1}^n M^{\xi}_i(\omega)}{n}$ exists and is equal to $M^{\xi}$.
For $i\geq 1$ we define $R_i$ to be a random variable giving the total time between the $i$-th and $i+1$-st visits to $\vzero$.
Note that $R_1$ is precisely the variable $R$ giving the return time to $\vzero$.
 By the strong law of large numbers, for almost all runs $\omega$ the limit $\lim_{n\rightarrow \infty} \frac{\sum_{i=1}^n R_i(\omega)}{n}$ exists and is equal to $\Ex{R_1}=\Ex{R}<\infty$.

Now note that for $t\geq 0$ and $n\geq 1$ satisfying $\sum_{i=1}^{n-1} R_i(\omega)\leq t\leq \sum_{i=1}^n R_i(\omega)$ (here for $n=1$ we assume $\sum_{i=1}^{n-1} R_i(\omega)=0$) we have
\begin{equation}\label{eq:three}
\frac{\sum_{i=1}^{n-1} M^{\xi}_i(\omega)}{\sum_{i=1}^n R_i(\omega)}\quad \leq \quad\frac{O^{\leq t}_{\xi}(\omega)}{t}\quad
\leq \quad\frac{\sum_{i=1}^{n} M^{\xi}_i(\omega)}{\sum_{i=1}^n R_i(\omega)}
\end{equation}
The limit of the rightmost term of (\ref{eq:three}) exists:
\[
\frac{M^{\xi}}{\Ex{R}}\quad = \quad\lim_{n\rightarrow \infty} \frac{\sum_{i=1}^n M^{\xi}_i(\omega)}{n}\cdot \lim_{n\rightarrow \infty}\frac{n}{\sum_{i=1}^n R_i(\omega)}\quad = \quad \lim_{n\rightarrow \infty} \frac{\sum_{i=1}^n M^{\xi}_i(\omega)}{\sum_{i=1}^n R_i(\omega)}
\]
Similarly, the limit of the left most term of (\ref{eq:three}) exists:
\[
\frac{M^{\xi}}{\Ex{R}}\quad = \quad \left(\lim_{n\rightarrow \infty} \frac{\sum_{i=1}^{n-1} M^{\xi}_i(\omega)}{n-1}\cdot\lim_{n\rightarrow \infty} \frac{ n-1}{n}\right)\cdot \lim_{n\rightarrow \infty}\frac{n}{\sum_{i=1}^n R_i(\omega)}\quad = \quad \lim_{n\rightarrow \infty} \frac{\sum_{i=1}^n M^{\xi}_i(\omega)}{\sum_{i=1}^n R_i(\omega)}
\]
Thus by basic properties of limits, the limit
$\lim_{t\rightarrow \infty} \frac{O^{\leq t}_{\xi}(\omega)}{t}$ exists
and is equal to $\frac{M^{\xi}}{\Ex{R}}$.
\end{proof}
The following two claims show that the frequencies indeed solve the traffic LP with $\delta<1$.
\begin{claim}
For every $j\in \{1,\ldots,n\}$ we have
\begin{equation}\label{eq:flow}
\sum_{\xi\in \Sigma_j} O_{\xi} = \sum_{i=0}^n \sum_{\zeta\in \Sigma_i} O_{\zeta}\cdot A_{\zeta j}
\end{equation}
Here we assume (in order to simplify notation) that the external source is a (symbolic) queue $0$, with action set $\Sigma=\{\iota\}$ where $\iota$ has rate $\mu_0$ and a production function $\Prob_0$ (thus $O_{\iota}\cdot A_{\iota j}=\valpha_j$).
\end{claim}
\begin{proof}
Denote by $O^{\leq t}_{\zeta\rightarrow j}(\omega)$ the number of jobs produced for queue $j$ by firing the action $\zeta$. The crucial observation is that for almost all runs $\omega$ we have
\[
\lim_{t\rightarrow \infty} \frac{\sum_{\xi\in \Sigma_j} O^{\leq t}_{\xi}(\omega)}{t}\quad = \quad \lim_{t\rightarrow \infty}\frac{\sum_{i=0}^n \sum_{\zeta\in \Sigma_i} O^{\leq t}_{\zeta\rightarrow j}(\omega)}{t}
\]
Indeed, $\sum_{\xi\in \Sigma_j} O^{\leq t}_{\xi}(\omega)\leq \sum_{i=0}^n \sum_{\zeta\in \Sigma_i} O^{\leq t}_{\zeta\rightarrow j}(\omega)$ for all $t$ and $\omega$ because $\sum_{i=0}^n \sum_{\zeta\in \Sigma_i} O^{\leq t}_{\zeta\rightarrow j}(\omega)$ is precisely the number of jobs that enter the queue $j$ up to time $t$. On the other hand, $\sum_{i=0}^n \sum_{\zeta\in \Sigma_i} O^{\leq t}_{\zeta\rightarrow j}(\omega)\leq \sum_{\xi\in \Sigma_j} O^{\leq t}_{\xi}(\omega)$ for infinitely many $t$ and almost all $\omega$ because (almost surely) queue $j$ becomes empty infinitely many times.

We obtain that
\begin{eqnarray*}
\sum_{\xi\in \Sigma_j} O_{\xi} & = & \sum_{\xi\in \Sigma_j} \lim_{t\rightarrow{} \infty} \frac{O^{\leq t}_{\xi}(\omega)}{t} \\
 & = & \lim_{t\rightarrow{} \infty} \frac{\sum_{\xi\in \Sigma_j} O^{\leq t}_{\xi}(\omega)}{t} \\
 %& = & \lim_{t\rightarrow{} \infty} \frac{O^{\leq t}_{i}(\omega)}{t} \\
 %& = & \lim_{t\rightarrow{} \infty} \frac{I^{\leq t}_{i}(\omega)}{t} \\
 & = & \lim_{t\rightarrow{} \infty} \frac{\sum_{i=0}^n \sum_{\zeta\in \Sigma_i} O^{\leq t}_{\zeta\rightarrow j}(\omega)}{t} \\
 & = & \sum_{i=0}^n \sum_{\zeta\in \Sigma_i}\lim_{t\rightarrow{} \infty} \frac{O^{\leq t}_{\zeta\rightarrow j}(\omega)}{t} \\
 & = & \sum_{i=0}^n \sum_{\zeta\in \Sigma_i}
 \lim_{t\rightarrow{} \infty} \frac{O^{\leq t}_{\zeta}(\omega)}{t}\cdot
 \lim_{t\rightarrow{} \infty}  \frac{O^{\leq t}_{\zeta\rightarrow j}(\omega)}{O^{\leq t}_{\zeta}(\omega)}  \\
 & = & \sum_{i=0}^n \sum_{\zeta\in \Sigma_i} O_{\zeta}\cdot A_{\zeta j}
\end{eqnarray*}
Here the last equality follows from Claim~1 and the fact that
$\lim_{t\rightarrow{} \infty}  \frac{O^{\leq t}_{\zeta\rightarrow j}(\omega)}{O^{\leq t}_{\zeta}(\omega)}$ is
the average number of jobs for queue $j$ produced by firing $\zeta$, i.e., $A_{\zeta j}$.
\end{proof}
\begin{claim}
For every $i\in \{1,\ldots,n\}$ we have
\begin{equation}\label{eq:cap}
\frac{\sum_{\xi\in \Sigma_i} O_{\xi}}{\mu_{i}}<1
\end{equation}
\end{claim}
\begin{proof}
Intuitively, we show that $\frac{\sum_{\xi\in \Sigma_i} O_{\xi}}{\mu_{i}}+\pi(\vzero)$ is equal to the proportion of time in which either the queue $i$ operates a job or all queues are empty. As this proportion is at most $1$ and $\pi(\vzero)>0$, we obtain the desired result.

Formally, denote by $T_n$ the total amount of time in which the first $n$ jobs are processed by queue $i$ (by this we mean the total time from the beginning of the computation needed to finish $n$ jobs using the ``machine'' operating queue $i$). Denote by $N_j$ the total time it takes to process the $j$-th job from queue $i$ (by this we mean the time the ``machine'' actively spends by processing the job).

Observe that for almost all runs $\omega$ we have that $\lim_{n\rightarrow \infty}
\frac{\sum_{j=1}^{n} N_j(\omega)}{n}$ is the average processing time of the queue $i$ which is $\frac{1}{\mu_i}$, and that
$\frac{n}{T_n(\omega)}$ is equal to
\[
\frac{\sum_{\xi\in \Sigma_i} O^{\leq T_n(\omega)}_{\xi}}{T_n(\omega)}
\]
which means that by Claim~1, % \stefan{Is the last equality by Claim (1)?}
\[
\lim_{n\rightarrow \infty}\frac{n}{T_n(\omega)}\quad =\quad
\lim_{n\rightarrow \infty} \frac{\sum_{\xi\in \Sigma_i} O^{\leq T_n(\omega)}_{\xi}}{T_n(\omega)}\quad = \quad \sum_{\xi\in \Sigma_i} O_{\xi}
\]
Thus
\begin{equation}\label{eq:util}
\lim_{n\rightarrow \infty}
\frac{\sum_{j=1}^{n} N_j(\omega)}{T_n(\omega)}\quad = \quad \lim_{n\rightarrow \infty}
\frac{\sum_{j=1}^{n} N_j(\omega)}{n}\cdot \lim_{n\rightarrow \infty}\frac{n}{T_n(\omega)}\quad = \quad
\frac{\sum_{\xi\in \Sigma_i} O_{\xi}}{\mu_i}
\end{equation}
Denote by $E_n$ the total amount of time the network spends in state $\vzero$ before queue~$i$ finishes $n$ jobs (by this we mean the sum of all time intervals up to time $T_n$ in which all queues are simultaneously idle before $n$ jobs are finished by the ``machine'' operating queue $i$).
Intuitively, $\lim_{n\rightarrow \infty} \frac{E_n(\omega)}{T_n(\omega)}$ is equal to the proportion of time spent in $\vzero$ which is equal to $\pi(\vzero)>0$. More precisely, by the Ergodic Theorem, see~Theorem 3.8.1~of~\cite{No98}, for almost all runs $\omega$ we have
\begin{equation}\label{eq:invzero}
0\quad <\quad \pi(\vzero)\quad = \quad \lim_{t\rightarrow \infty} \frac{1}{t} \int_0^t I[\vx(\omega)=\vzero](s) ds\quad = \quad
\lim_{n\rightarrow \infty} \frac{E_n(\omega)}{T_n(\omega)}
\end{equation}
Finally, by (\ref{eq:util}) and (\ref{eq:invzero}),
{
\begin{align*}\label{eq:proportion}
\frac{\sum_{\xi\in \Sigma_i} O_{\xi}}{\mu_i}\ & < \
\frac{\sum_{\xi\in \Sigma_i} O_{\xi}}{\mu_i}+\pi(\vzero)\ =\ \\ &
 \frac{\sum_{j=1}^{n} N_j(\omega)}{n}\cdot \lim_{n\rightarrow \infty}\frac{n}{T_n(\omega)}+
\lim_{n\rightarrow \infty}\frac{E_n(\omega)}{T_n(\omega)}
\ = \
\lim_{n\rightarrow \infty}
\frac{\sum_{j=1}^{n} N_j(\omega)+E_n(\omega)}{T_n(\omega)}
\ \leq \ 1
\end{align*}
}
\end{proof}
%
%Define $T_n$ to be the total amount of time for which the $i$-th queue is non-empty for $n$-th time (i.e. $T_n$ is the total amount of time which elapses from the moment when the queue $i$ becomes non-empty for the $n$-time to the moment when it becomes empty again). By
%
%
%Define $T_0,T_1,\ldots$ returning time instants (i.e. real numbers such that $T_0<T_1<\cdots$) as follows: $T_0=0$. $T_{n+1}$ is defined as follows:
%\begin{itemize}
%\item If in time $T_n$ the $i$-th queue is empty, then $T_{n+1}$ is the first time instant after $T_n$ when the $i$-th queue becomes non-empty.
%\item If in time $T_n$ the $i$-th queue is non-empty, then $T_{n+1}$ is the first time instant after $T_n$ when the $i$-th queue fires.
%\end{itemize}
%Define $N_0,N_1,\ldots$ as follows: $N_n=1$ if the $i$-th queue is non-empty at time $T_n$, otherwise $N_n=0$.
%%
%Observe that
%\[
%\lim_{n\rightarrow \infty} \frac{\sum_{k=0}^n N_k\cdot T_k}{\sum_{k=0}^n N_k}
%\]
%is the average time to fire the $i$-th queue whenever it is non-empty, which is clearly equal to $\frac{1}{\mu_i}$.
%Also, note that by the definition of limit,
%\[
%O_i\quad =\quad \lim_{n\rightarrow \infty} \frac{O^{\leq t}}{t}
%\quad = \quad \lim_{n\rightarrow \infty} \frac{\sum_{k=0}^n N_k}{\sum_{k=0}^n T_k}
%\]
%Thus
%\[
%O_i\quad =\quad \lim_{n\rightarrow \infty} \frac{\sum_{k=0}^n N_k}{\sum_{k=0}^n T_k}\quad <\quad \lim_{n\rightarrow \infty} \frac{\sum_{k=0}^n N_k}{\sum_{k=0}^n N_k\cdot T_k}\quad =\quad \mu_i
%\]
\begin{claim}
For every $i\in \{1,\ldots,n\}$, $\rho_i=\frac{\sum_{\xi\in \Sigma_i} O_{\xi}}{\mu_i}$.
\end{claim}
\begin{proof}
By the Ergodic Theorem, see~Theorem 3.8.1~of~\cite{No98}, and by~(\ref{eq:util}), for almost all runs $\omega$ we have
\begin{align*}
\rho_i\quad = \quad \lim_{t \to \infty} & \Pr(\vx_i(t) \ne 0)\quad  = \quad \sum_{\vx\in \N^n:\vx_i\not =0}  \pi(\vx)\quad \\ & = \quad
\lim_{t\rightarrow \infty}  \frac{1}{t} \int_0^t I[\vx_i(\omega)\not = 0](s) ds
\quad  = \quad \lim_{n\rightarrow \infty}
\frac{\sum_{j=1}^{n} N_j(\omega)}{T_n(\omega)}\quad =\quad\frac{\sum_{\xi\in \Sigma_i} O_{\xi}}{\mu_i}
\end{align*}
\end{proof}

\noindent This finishes a proof of Lemma~\ref{prop:erg-traffic}.

}{}

\end{document}